\documentclass[letterpaper,twocolumn,10pt]{article}
\usepackage{usenix,epsfig,endnotes}

	\usepackage{graphicx}
	\usepackage{balance} 
	\usepackage{amsthm, amssymb, amsmath}
	\usepackage{multirow, rotating, wasysym, url}
	\usepackage[tight]{subfigure}
	\usepackage[retainorgcmds]{IEEEtrantools}
	\usepackage[ruled,linesnumbered,vlined]{algorithm2e}
	\usepackage{hyperref} 
	\usepackage{color}
    \usepackage[normalem]{ulem}
    \usepackage{cases}

	\newcommand{\presec}{\vspace{-0.05in}}
	\newcommand{\postsec}{\vspace{-0.05in}}
	\newcommand{\presub}{\vspace{-0.05in}}
	\newcommand{\postsub}{\vspace{-0.05in}}

	\newcommand{\prefigcaption}{\vspace{-0in}}
	\newcommand{\postfig}{\vspace{-0.05in}}

		\mathchardef\Gamma="0100 \mathchardef\Delta="0101
\mathchardef\Theta="0102 \mathchardef\Lambda="0103
\mathchardef\Xi="0104 \mathchardef\Pi="0105
\mathchardef\Sigma="0106 \mathchardef\Upsilon="0107
\mathchardef\Phi="0108 \mathchardef\Psi="0109
\mathchardef\Omega="010A

\newcommand{\outline}[1]{}

\usepackage{cite}
\usepackage{xspace}
\usepackage{url}
\usepackage{graphicx}
\usepackage{latexsym}
\usepackage{amssymb}
\usepackage{amsfonts}
\usepackage{psfrag}
\usepackage{wrapfig}
\usepackage{comment}
\usepackage{alltt}
\usepackage{color}

\newtheorem{thrm}{Theorem}[section]

\newtheorem{thm}{Theorem}

\newcommand{\ie}{\emph{i.e.}\xspace}

\newcommand{\etc}{\emph{etc.}\xspace}
\newcommand{\etal}{\frenchspacing{}\emph{et al{.}}\xspace}

\newcommand{\Comment}[1]{}

\title{SF-sketch: A Two-stage Sketch for Data Streams}

\author{Tong Yang$^1$, Lingtong Liu$^2$, Yibo Yan$^1$, Muhammad Shahzad$^3$, Yulong Shen$^2$ \\ Xiaoming Li$^1$, Bin Cui$^1$, Gaogang Xie$^4$\\
$^1$Peking University, China. $^2$Xidian University, China. \\ $^3$North Carolina State University, USA.~~~~~ $^4$ICT, CAS, China
}

\begin{document}
\maketitle
\sloppy
	\begin{abstract}
A \emph{sketch} is a probabilistic data structure used to record frequencies of items in a multi-set.
Sketches are widely used in various fields, especially those that involve processing and storing data streams.
In streaming applications with high data rates, a sketch ``fills up'' very quickly.
Thus, its contents are periodically transferred to the remote collector, which is responsible for answering queries.
In this paper, we propose a new sketch, called Slim-Fat (SF) sketch, which has a significantly higher accuracy compared to prior art, a much smaller memory footprint, and at the same time achieves the same speed as the best prior sketch.
The key idea behind our proposed SF-sketch is to maintain two separate sketches: a small sketch called Slim-subsketch and a large sketch called Fat-subsketch.
The Slim-subsketch is periodically transferred to the remote collector for answering queries quickly and accurately.
The Fat-subsketch, however, is not transferred to the remote collector because it is used only to assist the Slim-subsketch during the insertions and deletions and is not used to answer queries.
We implemented and extensively evaluated SF-sketch along with several prior sketches and compared them side by side.
Our experimental results show that SF-sketch outperforms the most widely used CM-sketch by up to 33.1 times in terms of accuracy.
We have released the source codes of our proposed sketch as well as existing sketches at Github \cite{opensource}.
The short version of this paper will appear in ICDE 2017\cite{short}.
\end{abstract}

	\presec
\section{Introduction} \postsec

\subsection{Background and Motivation} \postsub
A \emph{sketch} is a probabilistic data structure that is used to record frequencies of distinct items in a multi-set.
Due to their small memory footprints, high accuracy, and fast speeds of queries, insertions, and deletions, sketches are being extensively used in data stream processing \cite{sdm10sketch}, \cite{chen2010tracking}, \cite{cormode2005sketching},\cite{cormode2008finding}, \cite{sigsketch}, \cite{asketch}, \cite{icde09fcm}, \cite{shbf}, \cite{sail},.
Sketches are also being applied in many other fields, such as NLP \cite{cubest,NLPstreaming09,NLPsketch11},
 sparse approximation in compressed sensing \cite{gilbert2007one}, and more \cite{li1sketch,durme2009streaming,goyalIII}.
%
%

Data streams are generated and used in many scenarios, such as network traffic, graph streams, and multi-media streams.
For example, a typical application is to measure the number of packets for each flow in the traffic of a network.
In most of these applications, the rate at which the data arrives in the stream is very high and a sketch being used to record information in a stream ``fills up'' very quickly.
Consequently, each monitoring node, such as a router or switch, that populates the sketch has to periodically transfer the contents of the filled up sketch to some remote collector \cite{flowradar,remotecollectorsigcomm,remotecollectorNSDI}, which stores the contents of these sketches and answers any queries.

As existing sketches need only a few (\textit{e.g.}, 4) memory accesses for each insertion, the insertion speed of sketches is often fast enough to keep up with the fast rates of data streams.
However, the bottleneck lies at the bandwidth and high speed memory of remote controller, because it receives filled up sketches from a large number of monitoring nodes.
Consequently, it is critical that the size of the filled up sketch that each monitoring node sends to the collector be very small but still contain enough information that allows the remote collector to answer queries accurately.
The accuracy of a sketch in answering queries quantifies how close the value of the frequency estimated from the sketch is to the actual value of the frequency.
This paper focuses on the design of a new scheme that not only significantly reduces the size of the sketch that it sends to the collector compared with the existing sketches, but is also more accurate while achieving the same query speed as the best prior sketches.

\presub
\subsection{Limitations of Prior Art} \postsub
Here we only introduce four typical sketches: C-sketch, CM-sketch, CU-sketch, and CML-sketch.
Charikar \etal proposed the Count sketch (C-sketch) \cite{countsketch}.
C-sketch experiences two types of errors: over-estimation error, where the result of a query is a value larger than the true value, and under-estimation error, where the result of a query is a value smaller than the true value.
Improving on the C-sketch, Cormode and Muthukrishnan proposed the Count-min (CM) sketch \cite{cmsketch}, which does not suffer from the under-estimation error, but only from the over-estimation error. And the authors claimed that such one-side error has many benefits \cite{cmsketch}.
In a further enhancement, Cormode \etal proposed the conservative update (CU) sketch \cite{cusketch}, which improves the accuracy at the cost of not supporting item deletions. 
CML-sketch \cite{cmlog} further improves the accuracy at the cost of suffering from both over-estimation and under-estimation errors.
Because CM-sketch supports deletions and does not have under-estimation errors, it is still the most popular sketch in practice.
As mentioned above, given the desired accuracy and transmission period, smaller sketch leads to less requirement for bandwidth and fast expensive fast memory and faster query speed in the collector.
%

\presub
\subsection{Proposed Approach}
\label{subsec:ProposedApproach}
\postsub
In this paper, we present a novel two-stage sketch, called Slim-Fat (SF) sketch.
The key idea behind our proposed SF-sketch is to maintain two separate sketches: a small sketch called Slim-subsketch and a large sketch called Fat-subsketch.
Both subsketches have similar accuracy, and thus, we only need to send the slim-subsketch to the remote collector.
Compared to the state-of-the-art, the slim-subsketch achieves significantly higher accuracy and significantly smaller memory footprint while supporting deletions and achieving the same query speed.
%

%
Before describing our approach, we first briefly describe how the conventional CM-sketch works because several design choices in SF-sketch are built on it.
%
%
%
As shown in Figure \ref{draw:cmsketch}, a CM-sketch consists of $d$ arrays, where we represent the $i^\text{th}$ array with $\textbf{A}_i$.
Each array consists of $w$ buckets and each bucket contains one counter.
We represent the counter in the $j^\text{th}$ bucket of the $i^\text{th}$ array with $A_i[j]$.
Each array $\textbf{A}_i$ ($1\leqslant i\leqslant d$), is associated with an independent hash function $h_i(.)$, whose output is uniformly distributed in the range $[1, w]$. 
In CM-sketch, the initialization operation is simply to set all the $d*w$ counters to zero.
To insert an item $e$, \ie, to increment its frequency stored in the sketch by 1, the CM-sketch computes the $d$ hash functions $h_1(e), h_2(e), \ldots, h_d(e)$ and increments the $d$ counters $A_1[h_1(e)], A_2[h_2(e)], \ldots, A_d[h_d(e)]$ by 1.
For convenience, we call them \texttt{the $d$ hashed counters}.
%
To delete an item $e$, the CM-sketch computes $d$ hash functions and decrements the $d$ hashed counters by 1.
When querying the frequency of an item $e$, the CM-sketch computes hash functions and returns the value of the smallest one among the $d$ hashed counters.
Note that the value returned by the CM-sketch in response to the query for the frequency of an item will never be smaller than the true value of its frequency.
Consequently, CM-sketch suffers only from the under-estimation error, but not from the over-estimation error \cite{cmsketch}.
%
%
%
%

\begin{figure}[htbp]
\postfig
\centering
\includegraphics[width=0.4\textwidth]{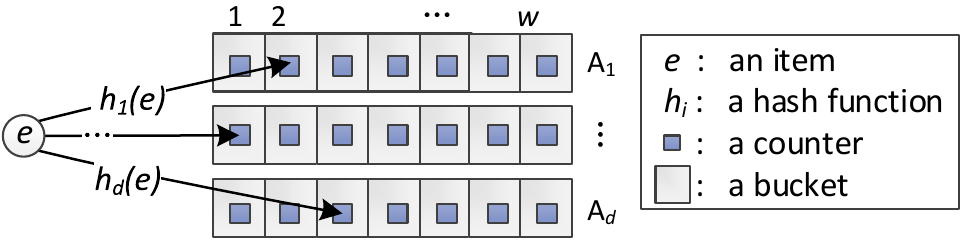}
\postfig\postfig
\caption{The Count-min sketch architecture.}
\label{draw:cmsketch}
\postfig
\end{figure}

In SF-sketch, the Slim-subsketch, as the name suggests, has significantly fewer counters compared to the Fat-subsketch.
The motivation behind keeping the small Slim-subsketch is to use smaller memory while keeping high accuracy.
The motivation behind keeping the relatively large Fat-subsketch is to assist the Slim-subsketch during updates so as to make the accuracy of the Slim-subsketch as high as possible.
Fat-subsketch uses many more counters compared to Slim-subsketch, but only needs tens of Mega bytes, which is very small compared to the large but cheap memory of current computers.
In designing our SF-sketch, we start with a bare bones version of the sketch and make improvements step by step to arrive at its final design.
In this process, we present five versions of SF-sketch, namely SF$_1$-sketch through SF$_4$-sketch, and SF$_\text{F}$-sketch, where each version improves upon the previous versions by addressing some of the limitations in those versions.
The version SF$_\text{F}$-sketch is the final design of our SF-sketch.

In SF$_1$-sketch, the Slim-subsketch consists of $d\times w$ counters and the Fat-subsketch consists of $d\times w'$ counters, where $w'>w$.
The Fat-subsketch in SF$_1$-sketch is the same as a standard CM-sketch.
The \textit{key idea} behind the design of SF$_1$-sketch is that, when inserting an item $e$, \textit{if the value of a counter to which the hash function $h_i(e)$ points is already greater than the real frequency of the item $e$, then incrementing that counter will only degrade the accuracy.}
Specifically, when inserting an item $e$, we first insert it into the Fat-subsketch using the insertion operation of the standard CM-sketch and then query its current frequency from this Fat-subsketch using the query operation of the standard CM-sketch.
Suppose the Fat-subsketch estimates $e$'s current frequency to be $c$.
Next, we compute $d$ hash functions corresponding to the $d$ arrays of the Slim-subsketch, and retrieve $d$ hashed counters.
If all $d$ counters are less than or equal to the estimate $c$, we increment only the smallest counter(s)
by 1 in the Slim-subsketch; \textit{otherwise, we do nothing.}
The motivation behind incrementing no counter or only the smallest counter(s) in the Slim-subsketch is twofold.
First, by reducing the number of counters that we increment, the over-estimation error reduces.
Second, as fewer counters are incremented, their sizes can be reduced, which reduces the footprint of the Slim-subsketch. Unfortunately, the SF$_1$-sketch does not support deletions.
To well support deletion, we propose several optimization techniques, and finally arrive at the final version -- SF$_\text{F}$ sketch.

\presub
\subsection{Technical Challenges} \postsub
%
%
%
The first technical challenge is to achieve a significantly higher accuracy compared to the CM-sketch, which is currently the most widely used sketch.
%
To address this challenge, we leverage our novel insight that if we reduce the number of counters that are incremented for each insertion, the accuracy will be improved because the extent of over-estimation error will decrease.
When inserting a new item, our proposed sketch does not always increment $d$ counters in the Slim-subsketch, \textit{rather increments no counter or only the smallest counter(s)} to avoid under-estimation error.
Note that the query is only processed based on the information stored in the Slim-subsketch, which is why we focus on minimizing the number of counter increments per insertion only in the Slim-subsketch.
%
To determine exactly which counters to increment in the Slim-subsketch, our SF-sketch makes use of the Fat-subsketch, which enables it to estimate the number of times the item has already been inserted.
SF-sketch then either increments only the smallest counters in the Slim-subsketch if the value of the smallest counters is less than this estimated value, or increments no counter at all.

The second technical challenge is to enable Slim-subsketch to support deletions.
It is very difficult to achieve accurate deletions from the Slim-subsketch because to support deletions, one needs to keep track of exactly which counters were incremented when each item was inserted.
This information is required to identify the appropriate counters to decrement when deleting an item and to identify the influence of those decrements on other items.
Tracking such information is very expensive, both in terms of memory overhead and computational cost.
To address this challenge, instead of achieving accurate deletions, \ie, decrementing all those counters that were incremented at the time of inserting the given item, we achieve approximate deletions, \ie, decrementing as many counters in the Slim-subsketch as possible without causing any under-estimation errors.

\presub
\subsection{Key Contributions} \postsub
\begin{enumerate}
\item We propose a new sketch, namely the SF-sketch, which has higher accuracy compared to the prior art while supporting deletions and keeping the query speed unchanged.
\item We implemented C-sketch, CM-sketch, CU-sketch, CML-sketch and SF-sketch on GPU and multi-core CPU platforms. We carried out extensive experiments on these two platforms to evaluate and compare the performance of all these sketches. Experimental results show that SF-sketch outperforms CM-sketch by up to 33.1 times in terms of average relative error.
\end{enumerate}

\section{Related Work} 
\label{sec:relatedwork}

The structure of the Count sketch (C-sketch) \cite{countsketch} proposed by Charikar \etal is exactly the same as the CM-sketch \cite{cmsketch} described earlier except that each array $\textbf{A}_i$ is associated with two hash functions $h_i(.)$ and $\delta_i(.)$.
%
%
Each hash function $\delta_i(.)$ evaluates to -1 or +1 with equal probability.
%
%
%
%
To insert an item $e$, for all values of $i\in[1, w]$, C-sketch calculates hash functions $h_i(e)$ and $\delta_i(e)$ and adds $\delta_i(e)$ to the counters $A_i[h_i(e)]$.
When querying the frequency of item $e$, C-sketch reports the median of $\{A_1[h_1(e)] \times \delta_1(e), A_2 [h_2 (e)] \times \delta_2(e) \ldots A_d [h_d(e)] \times \delta_d(e)\}$.
%

Unfortunately, C-sketch suffers from both over-estimation and under-estimation errors.
Therefore, several improvements, which do not suffer from the under-estimation errors, have been proposed
such as the CM-sketch \cite{cmsketch}, CU-sketch \cite{cusketch}\footnote{Estan and Varghese proposed the CU-sketch which can be combined with other sketches. For convenience, CU-sketch means CM-CU-sketch when it is combined with CM sketch.}, and Count-Min-Log (CML) sketch \cite{cmlog}.
These three sketches all have $d$ arrays of $w$ counters each.
To insert an item $e$, CU-sketch \cite{cusketch} increments only the smallest counter(s) among the $d$ hashed counters.
Although CU-sketch improves the query accuracy significantly, its fundamental limitation is that it does not support deletions, and consequently it has not received as wide acceptance in practice as the CM-sketch.
%
%
%
%
CML-sketch is another variant of the CM-sketch that uses logarithm-based approximate counters instead of linear counters \cite{cmlog}.
Instead of incrementing one counter per array per insertion, it decides whether or not to increment the counters each time with logarithmic probabilities.
This helps in reducing the number of bits for each counter, which in turn allows the sketch to have more counters in the same amount of memory and thus achieve better accuracy.
Unfortunately, CML-sketch suffers from both over-estimation and under-estimation errors, and its final version does NOT support deletions. Thorough statistical analysis of various sketches is provided in \cite{sketchtheorysigmod,sketchtheorytods}.

A recent work presented Augment sketch (A-sketch), which is a universal framework that can be applied to many existing sketches, especially to those with low accuracy \cite{asketch}.
A-sketch uses a filter to catch heavy hitters (high-frequency items) earlier, and uses classical sketches (such as CM-sketch and C-sketch) to store and query the rest items.
In this way, the accuracy could be improved.
However, always keeping the most frequent items in the first filter without incurring additional errors is a challenging issue.
Complex design and frequent communications between the two filters are unavoidable, making the implementation complicated.
Indeed, A-sketch can be applied to our SF-sketch as well.
However, according to our test results, as our SF-sketch is already very accurate, combining A-sketch with SF-sketch brings little increase in accuracy but does bring more complexity.
%
%

Another class of data structures that can be used to store frequencies of items are the enhanced Bloom filters, such as Spectral Bloom Filters (SBF) \cite{spectralBF}, Dynamic Count Filters (DCF) \cite{dynamicCF}, and more \cite{shbf}, which indeed can estimate frequencies of items.
SBF replaces each bit in the conventional Bloom filter with a counter \cite{spectralBF}.
To insert an item, the basic version of SBF simply increments all the counters that the item maps to.
When querying the frequency of an item, SBF returns the value of the smallest counter(s) among all the counters to which the hash functions map the item to as the estimate of the frequency of that item in the multiset.
%
%
DCF extends the concept of SBF while improving the memory efficiency of SBF by using two separate filters \cite{dynamicCF}.
The first filter is comprised of fixed size counters while the size of counters in the second filter is dynamically adjusted.
The use of two filters, unfortunately, increases the complexity of DCF, which degrades its query and update speeds.

	\presec
\section{The Slim-Fat Sketch} \postsec
\label{sec:mcsketch}
In this section, we present the details of our SF-sketch.
To better explain the intuition at work behind the SF-sketch and to justify the design choices we made in developing the SF-sketch, we will start with a basic version and improve it incrementally to arrive at the final design.
For each intermediate version of the SF-sketch that we develop while working our way towards the final design, we will first describe its insertion, query, and deletion operations.
After that we will discuss its limitations, which will guide us in making our design choices for the next version.
In this process, we will present five different versions of SF-sketch, which we name SF$_1$-sketch through SF$_\text{4}$-sketch, and finally SF$_\text{F}$-sketch, which is our final design.
Each version is developed by studying the limitations of its predecessor version and addressing them.
%
%

\begin{figure}[htbp]
\centering
\includegraphics[width=0.45\textwidth]{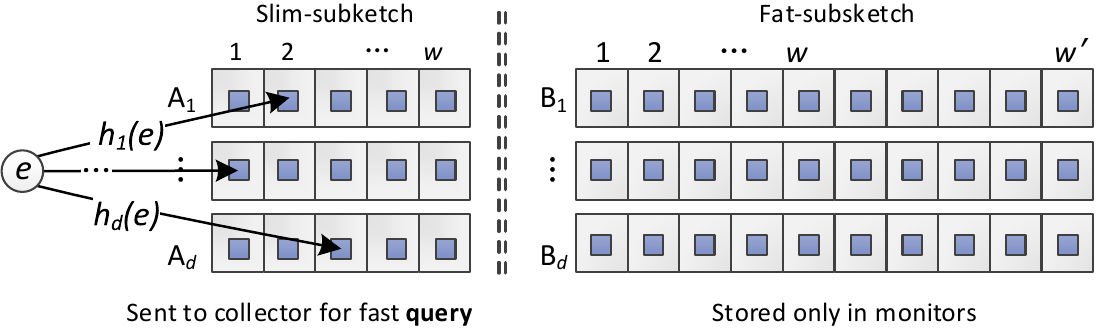}
\prefigcaption
\caption{The Slim-Fat sketch architecture.}
\label{draw:fat-sketch}
\postfig \postfig
\end{figure}

\noindent\textbf{Rationale:}
In our slim-fat architecture (shown in Figure \ref{draw:fat-sketch}), there is a set of arrays with fewer counters per array called a Slim-subsketch, and a set of arrays with comparatively more counters per array called a Fat-subsketch.
When inserting or deleting an item, we first update the Fat-subsketch, and then update the Slim-subsketch based on the observations we make from the Fat-subsketch.
The key insight at work behind our proposed scheme is that, \textcolor{blue}{\textit{when inserting an item,
if the smallest one of the $d$ hashed counters is already bigger than its current true frequency, then incrementing any counter only degrades the accuracy.} }
%
%
\textcolor{blue}{As the true accuracy is not easy to obtain using small memory, we use Fat-subsketch to give a relatively accurate estimate of the current true frequency.}
%
Next, we start with the first version of our slim-fat sketch, \ie, the SF$_1$-sketch, and discuss its operations and limitations, which will pave the way towards the design of SF$_2$-sketch and its subsequent versions.
Table~\ref{table:symbols} summarizes the symbols and abbreviations used in this paper.

\begin{table}[htbp]
\centering
\caption{Symbols \& abbreviations used in the paper}
\begin{tabular}{|c|l|}
\hline
\textbf{Symbol}&\textbf{Description}\\
\hline
$e$& Any item that can be handled by SF$_i$-sketch\\
\hline
$d$& \# of arrays in Slim-subsketch and Fat-subsketch\\
\hline
\multirow{2}{*}{$z$} & \# of counters in each bucket of Fat-subsketch of \\ & SF$_4$- and SF$_\text{F}$-sketch\\
\hline
\multirow{2}{*}{$w / w'$}& \# of counters or buckets in each array of \\& Slim- / Fat-subsketch\\
\hline
$\textbf{A}_i$& the $i^\text{th}$ array in the Slim-subsketch of SF$_i$-sketch\\
\hline
$\textbf{B}_i$& the $i^\text{th}$ array in the Fat-subsketch of SF$_i$-sketch\\
\hline
\multirow{2}{*}{$\textbf{C}_i$}& the $i^\text{th}$ array in the Deletion-subsketch\\ & used in SF$_2$-sketch\\
\hline
\multirow{2}{*}{$h_i(e)$} & the $i^{th}$ hash function used in Slim- and \\& Deletion-subsketch\\
\hline
$g_i(e)$& the $i^{th}$ hash function used in Fat-subsketch\\
\hline
\multirow{2}{*}{$B_e^{\min}$}&  the minimum value among all counters in \\&  \{$\textbf{B}_i \lvert 1\leqslant i\leqslant d$\} \\
\hline
$\%$& mod operation\\
\hline
\end{tabular}
\label{table:symbols}
\vspace{-0.13in}
\end{table}

\presub
\subsection{SF$_1$: Optimizing Accuracy Using One fat-Subsketch}
\postsub
\label{sec:SF1}

As shown in Figure \ref{draw:fat-sketch}, SF$_1$-sketch consists of $d$ arrays in both the Slim-subsketch and the Fat-subsketch.
\textit{The Fat-subsketch is exactly a standard CM-sketch with many more counters than the Slim-subsketch.}
We represent the $i^\text{th}$ array in the Slim-subsketch with $\textbf{A}_i$ and in the Fat-subsketch with $\textbf{B}_i$.
Each array in the Slim-subsketch consists of $w$ buckets while each array in the Fat-subsketch consists of $w'$ buckets, where $w'>w$.
Furthermore, each bucket in both Slim and Fat-subsketches contains one counter.
We represent the counter in the $j^\text{th}$ bucket of the $i^\text{th}$ array in the Slim-subsketch with $A_i[j]$, where $1\leqslant i\leqslant d$ and $1\leqslant j\leqslant w$.
Similarly, we represent the counter in the $k^\text{th}$ bucket of the $i^\text{th}$ array in the Fat-subsketch with $B_i[k]$, where $1\leqslant i\leqslant d$ and $1\leqslant k\leqslant w'$.
Each array $\textbf{A}_i$ is associated with a uniformly distributed independent hash function $h_i(.)$, where the output of $h_i(.)$ lies in the range $[1, w]$.
Similarly, each array $\textbf{B}_i$ is associated with a uniformly distributed independent hash function $g_i(.)$, where the output of $g_i(.)$ lies in the range $[1, w']$.
The structure of the SF$_1$-sketch is shown in Figure \ref{draw:fat-sketch}.
The \textbf{initialization} operation for the SF$_1$-sketch is simply setting all counters $A_i[j]$ and $B_i[k]$ to zero, where $1\leqslant i\leqslant d$, $1\leqslant j\leqslant w$, and $1\leqslant k\leqslant w'$.
%
%

\noindent\textbf{Insertion: }
When inserting an item, the SF$_1$-sketch first inserts it into the Fat-subsketch, and based on the observations made from the Fat-subsketch, increments appropriate counters in the Slim-subsketch.
\textit{The insertion operation in the Fat-subsketch is exactly the same as the conventional CM-sketch.}
To insert an item $e$ into the Fat-subsketch, we first compute the $d$ hash functions $g_1(e), g_2(e), \ldots, g_d(e)$ and increment the $d$ hashed counters $B_1[g_1(e)], B_2[g_2(e)], \ldots, B_d[g_d(e)]$ by 1.
After inserting the item, we estimate its current frequency of $e$ by finding the minimum value among the $d$ hashed counters we just incremented and represent it with $B_e^{\min}$.
To insert the item $e$ into the Slim-subsketch, we compute the $d$ hash functions and identify the smallest counter(s) among the $d$ hashed counters $A_1[h_1(e)], A_2[h_2(e)], \ldots, A_d[h_d(e)]$. If the value of the smallest counter(s) are not smaller than $B_e^{\min}$, insertion operation ends.
Otherwise, we increment the smallest counter(s) by 1.
Note that CU-sketch always increments the smallest counter(s).
Thus SF$_1$-sketch is much more accurate than CU-sketch.
%
%
In other words, $\forall l\in[1,d]$, SF$_1$-sketch increment all counters $A_l[h_l(e)]$ by one that satisfy the following two conditions: $A_l[h_l(e)] = \min_{i=1}^{d} A_i[h_i(e)]$, and $ A_l[h_l(e)]<B_e^{\min}$.

\noindent\textbf{Query: }
When querying the frequency of item $e$, the SF$_1$-sketch computes the $d$ hash functions $h_1(e), h_2(e), \ldots, h_d(e)$, and returns the value of the smallest counter among $A_1[h_1(e)], A_2[h_2(e)],
\ldots, A_d[h_d(e)]$ as the result of the query.
Note that the query is answered only from the Slim-subsketch.

\noindent\textbf{Deletion: }
SF$_1$-sketch does not support deletions.
%


%
%
%

\noindent\textbf{Advantages and Limitations: }
The key advantage of the SF$_1$-sketch is that to answer a query it does not access the Fat-subsketch, but only accesses the Slim-subsketch, which keeps the query speed of this sketch as fast as the conventional CM-sketch.
Furthermore, note that during the insertion operation, we either increment no counters or increment only the smallest counter(s) in the Slim-subsketch.
%
The smallest counter in the Fat-subsketch gives the upper bound on the number of times that a given item has already been inserted.
This strategy reduces the number of increments in the Slim-subsketch, which has two advantages.
First, it reduces the memory footprint of the Slim-subsketch on the expensive and limited memory.
Second, due to fewer increments, the over-estimation error is reduced.
Unfortunately, the biggest limitation of the SF$_1$-sketch is that it does not support deletions from the Slim-subsketch.
While the Fat-subsketch assists the Slim-subsketch during insertion operation, it cannot assist in the deletion operation because the numbers of counters per array in the Fat- and Slim-subsketches are not the same.
This inability to support deletions from the Slim-subsketch limits the practical usability of the SF$_1$-sketch.
In the next version of our SF-sketch, \ie, the SF$_2$-sketch, we address this limitation while keeping the advantages of the SF$_1$-sketch.

\presub \presub
\subsection{SF$_2$: Supporting Deletion Deletion-subsketch}
\postsub

\begin{figure}[htbp]
\centering
\includegraphics[width=0.28\textwidth]{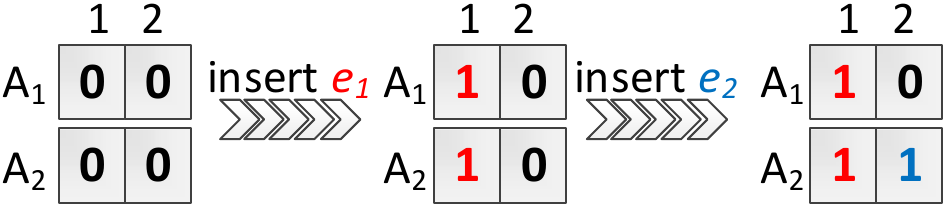}
\postfig
\caption{An example of the deletion problem.}
\label{draw:deletionsample}
\postfig
\end{figure}

\noindent\textbf{Difficulties for deletions:}
It is challenging to achieve accurate deletions in SF$_1$-sketch because to delete items from the Slim-subsketch of SF$_1$-sketch, one has to keep track of exactly which counters were incremented when inserting each item.
Such tracking is difficult and requires large memory and processing overhead.
We explain this with help of an example.
As shown in Figure \ref{draw:deletionsample}, consider a Slim-subsketch that has two arrays and two counters per array, where all counters are initialized to 0.
Let us first insert two items $e_1$ and $e_2$ and then delete the item $e_1$.
Furthermore, let $e_1$ maps to \textbf{A}$_1[1]$ and \textbf{A}$_2[1]$ and $e_2$ maps to \textbf{A}$_1[1]$ and \textbf{A}$_2[2]$.
In inserting $e_1$, we increment \textbf{A}$_1[1]$ and \textbf{A}$_2[1]$ both to 1.
After that, in inserting $e_2$, as the current value of \textbf{A}$_1[1]$ is 1 and \textbf{A}$_2[2]$ is 0, we only increment the smaller of the two, \ie, \textbf{A}$_2[2]$ to 1.
At this point, \textbf{A}$_1[1]=1$, \textbf{A}$_1[2]=0$, \textbf{A}$_2[1]=1$, and \textbf{A}$_2[2]=1$.
In deleting $e_1$, as $e_1$ maps to both \textbf{A}$_1[1]$ and \textbf{A}$_2[1]$ and as both were incremented at the time of inserting $e_1$, if we decrement them both, the query result of $e_2$ will be 0, \ie, an under-estimation error occurs, which we do not want in our SF-sketch.

\noindent\textbf{Deletion-subsketch:} To support deletions, in addition to one Slim-subsketch and one Fat-subsketch just like in the SF$_1$-sketch, the SF$_2$-sketch maintains another sketch, called the Deletion-subsketch.
The Deletion-subsketch is essentially a standard CM-sketch.
Unlike Fat-subsketch, all the parameters ($d$, $w$, $h_i(.)$) of the Deletion-subsketch and the Slim-subsketch are exactly the same.
For the Deletion-subsketch, we represent the counter in the $j^\text{th}$ bucket of the $i^\text{th}$ array with $C_i[j]$, where $1\leqslant i\leqslant d$ and $1\leqslant j\leqslant w$.
Note that the Fat-subsketch helps in deciding which counters to increment in the Slim-subsketch while inserting an item, whereas the Deletion-subsketch helps in deciding which counters to decrement in the Slim-subsketch when deleting an item.
The \textbf{initialization} operation for the SF$_2$-sketch consists of simply setting all counters $A_i[j]$, $B_i[k]$, and $C_i[j]$, to 0 ($1\leqslant i\leqslant d$, $1\leqslant j\leqslant w$, and $1\leqslant k\leqslant w'$.)

\noindent\textbf{Insertion: }
The insertion operation of the SF$_2$-sketch for the Slim- and Fat-subsketches is exactly the same as that of the SF$_1$-sketch, except that for the SF$_2$-sketch, we also add information about the incoming item to the Deletion-subsketch.
Specifically, to insert an item $e$ into the Deletion-subsketch, we compute $d$ hash functions and increment the $d$ hashed counters $C_1[h_1(e)], C_2[h_2(e)], \ldots, C_d[h_d(e)]$ by 1.

\noindent\textbf{Query: }
The query operation of the SF$_2$-sketch is exactly the same as the SF$_1$-sketch.

\noindent\textbf{Deletion: }
To delete an item $e$ from the SF$_2$-sketch, we first delete it from the Fat-subsketch by decrementing the $d$ counters $B_1[g_1(e)], B_2[g_2(e)], \ldots, B_d[g_d(e)]$ by 1 and then delete it from the Deletion-subsketch
by decrementing the $d$ counters $C_1[h_1(e)], C_2[h_2(e)], \ldots, C_d[h_d(e)]$ by 1.
Finally, we delete it from the Slim-subsketch.
%
We leverage the fact that before deleting the item from the Deletion-subsketch, \textit{each counter in the Slim-subsketch is always less than or equal to the corresponding counter in the Deletion-subsketch}, because when inserting an item, even if a counter in the Slim-subsketch to which the incoming item maps to is not incremented, the corresponding counter in the Deletion-subsketch is always incremented.
To delete the item $e$ from the Slim-subsketch, for each $i\in[1, d]$, we compare $A_i[h_i(e)]$ with $C_i[h_i(e)]$ and decrement $A_i[h_i(e)]$ by 1 only when $A_i[h_i(e)] > C_i[h_i(e)]$.

\noindent\textbf{Advantages and Limitations: }
The SF$_2$-sketch is advantageous over the SF$_1$-sketch because it supports deletions.
However, it is not efficient in terms of memory usage and update speed because it has to maintain an additional sketch, the Deletion-subsketch, to support deletions from the Slim-subsketch.
In the next version of the SF-sketch, \ie, the SF$_3$-sketch, we address this limitation while keeping the advantages of both SF$_1$- and SF$_2$-sketches.

\presub \vspace{-0.05in}
\subsection{SF$_3$: Combining Fat-subsketch with Deletion-subsketch}
\postsub

In SF$_3$-sketch, we get rid of the separate Deletion-subsketch, and modify the Fat-subsketch so that, in addition to insertions, it can assist deletions in the Slim-subsketch.
The Fat-subsketch in the SF$_3$-sketch is similar to the Fat-subsketch in the SF$_1$- and SF$_2$-sketches.
However, in the Fat-subsketch of SF$_3$-sketch, the number of buckets in each array is given by $w' = z\times w$, where $z$ is a positive integer.
In other words, the Fat-subsketch consumes $z$ times as much memory as the Slim-subsketch.
The structure of the Slim-subsketch in the SF$_3$-sketch is exactly the same as the Slim-subsketches in the SF$_1$- and SF$_2$-sketches.
However, the hash functions $h_i(.)$, where $1\leqslant i\leqslant d$, associated with the Slim-subsketch are now derived from the hash functions $g_i(.)$, where the output of $g_i(.)$ lies in the range $[1, z\times w]$.
More specifically,

\vspace{-0.1in}
\begin{equation}
h_i(.)=\Big(g_i(.)-1\Big)\%w+1
\label{equ:fx}
\end{equation} \vspace{-0.1in}

%

\noindent Consequently, the value of the hash function $h_i(.)$ always lies in the range $[1, w]$, where $w$ is the number of buckets per array in the Slim-subsketch.
Note also that calculating the hash function $h_i(.)$ from the hash function $g_i(.)$ using the equation above essentially \emph{associates} each counter $A_i[j]$ in the Slim-subsketch with $z$ counters $B_i[j], B_i[j+w], B_i[j+2w], \ldots, B_i[j+(z-1)w]$ in the Fat-subsketch.
Every time a counter in the Slim-subsketch is incremented, it is certain that one of its \emph{associated} $z$ counters in the Fat-subsketch is also incremented.
%
This further means that the value of a counter in the Slim-subsketch will always be less than or equal to the sum of values of all its associated counters in the Fat-subsketch.
%
%

\noindent\textbf{Insertion: }
When inserting an item $e$, the SF$_3$-sketch first inserts it into the Fat-subsketch.
For this we compute the $d$ hash functions $g_1(e), g_2(e), \ldots, g_d(e)$ and increment the $d$ counters $B_1[g_1(e)], B_2[g_2(e)], \ldots, B_d[g_d(e)]$ by 1.
Next, we find the minimum value among all these $d$ counters and represent it with $B_e^{\min}$.
To insert the item $e$ into the Slim-subsketch, we first compute the $d$ hash functions $h_1(e), h_2(e), \ldots, h_d(e)$ using Equation \ref{equ:fx} and then increment all counters $A_l[h_l(e)]$ that satisfy the conditions $A_l[h_l(e)]<B_e^{\min}$ and  $A_l[h_l(e)] = \min_{i=1}^{d} A_i[h_i(e)]$, where $l\in[1,d]$.
Note that if $\min_{i=1}^{d} A_i[h_i(e)]\geqslant B_e^{\min}$, we do nothing.
%

\noindent\textbf{Query: }
The query operation of SF$_3$-sketches is exactly the same as that of SF$_1$- and SF$_2$-sketches.

\noindent\textbf{Deletion: }
To delete an item from the SF$_3$-sketch, we first delete it from the Fat-subsketch and then from the Slim-subsketch.
To delete the item $e$ from the Fat-subsketch, we first calculate the $d$ hash functions $g_1(e), g_2(e), \ldots, g_d(e)$ and then decrement the $d$ counters $B_1[g_1(e)], B_2[g_2(e)], \ldots, B_d[g_d(e)]$ by 1.
To delete the item $e$ from the Slim-subsketch, we leverage the fact stated earlier that before deleting the item from the Fat-subsketch, \textit{the value of a counter in the Slim-subsketch is always less than or equal to the sum of values of all its associated counters in the Fat-subsketch}, because when inserting an item, even if a counter in the Slim-subsketch is not incremented, one of the associated counters in the Fat-subsketch is always incremented.
To delete the item $e$ from the Slim-subsketch, after deleting it from the Fat-subsketch, for each $i\in[1, d]$, we compare $A_i[h_i(e)]$ with $\sum_{m=0}^{z-1}B_i[h_i(e)+(m\times w)]$ and decrement $A_i[h_i(e)]$ by 1 if $A_i[h_i(e)] > \sum_{m=0}^{z-1}B_i[h_i(e)+(m\times w)]$.
Note that each value of $h_i(e)$ is calculated using Equation \eqref{equ:fx}.

\noindent\textbf{Advantages and Limitations: }
The advantage of SF$_3$-sketch over the SF$_2$-sketch is that it does not have to maintain a separate Deletion-subsketch.
Unfortunately, it is not efficient in terms of deletion speed because to delete an item, it needs $d\times z$ memory accesses to add the counters in each array of the Fat-subsketch.
In the next version of our SF-sketch, \ie, the SF$_4$-sketch, we address this limitation while keeping the advantages of all three previous versions of the SF-sketch.
\balance

\presub
\subsection{SF$_4$: Improving Deletion Speed}
\postsub

\begin{figure}[htbp]
\centering
\includegraphics[width=0.48\textwidth]{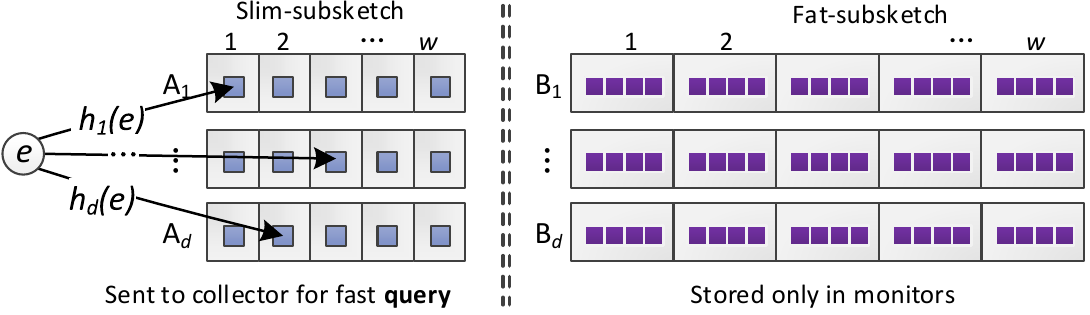}
\prefigcaption \postfig \postfig
\caption{The SF$_4$- and \& SF$_\text{F}$-sketch architecture.}
\label{draw:MCsketch}
\postfig
\end{figure}

In SF$_4$-sketch, we modify the Fat-subsketch so that instead of each bucket having one counter, each bucket has $z$ counters.
As shown in Figure \ref{draw:MCsketch}, in the Fat-subsketch of the SF$_4$-sketch, we have $d$ arrays with $w' = w$ buckets each, and each bucket now contains $z$ counters instead of one counter.
We represent the $k^\text{th}$ counter in the $j^\text{th}$ bucket of the $i^\text{th}$ array in the Fat-subsketch with $B_i[j][k]$, where $1\leqslant i\leqslant d$, $1\leqslant j\leqslant w$, and $1\leqslant k\leqslant z$.
Each array $\textbf{B}_i$ in the Fat-subsketch is associated with two uniformly distributed independent hash functions: $h_i(.)$ with output in the range $[1, w]$, which maps an item to a bucket in the $i^\text{th}$ array, and $f_i(.)$ with output in the range $[1, z]$, which maps an item to a counter inside the bucket $B_i[h_i(.)]$ of the $i^\text{th}$ array.
The Slim-subsketch uses the same has functions $h_i(.)$ as the Fat-subsketch to map items to buckets.
Every time a counter in the Slim-subsketch is incremented, it is certain that one of the counters among the $z$ counters in the corresponding bucket of the Fat-subsketch is also incremented.
This means that the value of a counter in the Slim-subsketch will always be less than or equal to the \textit{\textbf{sum}} \textit{of the values of all counters in the corresponding bucket} in the Fat-subsketch.
%
%
%

\noindent\textbf{Insertion: }
When inserting an item, the SF$_4$-sketch first inserts it into the Fat-subsketch, and based on the observations it makes from the Fat-subsketch, increments appropriate counters in the Slim-subsketch.
Specifically, to insert an item $e$ into the Fat-subsketch, we first compute $d$ hash functions $h_1(e), h_2(e), \ldots, h_d(e)$ and another $d$ hash functions $f_1(e), f_2(e), \ldots, f_d(e)$ and increment the $d$ counters $B_1[h_1(e)][f_1(e)]$, $B_2[h_2(e)][f_2(e)]$, $\ldots$, $B_d[h_d(e)][f_d(e)]$ by 1.
Next, we find the minimum value among all counters we just incremented and represent it with $B_e^{\min}$.
To insert the item $e$ into the Slim-subsketch, we identify the counters with the smallest value among the $d$ counters $A_1[h_1(e)], A_2[h_2(e)], \ldots, A_d[h_d(e)]$ and increment them by 1 only if their values are less than $B_e^{\min}$.
In other words, we increment all counters $A_l[h_l(e)]$ by one that satisfy the conditions $A_l[h_l(e)] = \min_{i=1}^{d} A_i[h_i(e)]$ and $ A_l[h_l(e)]<B_e^{\min}$, where $l\in[1,d]$.
If $\min_{i=1}^{d} A_i[h_i(e)]\geqslant B_e^{\min}$, we do nothing.
%
%
%

\noindent\textbf{Query: }
The query operation of the SF$_4$-sketch is exactly the same as the SF$_1$-, SF$_2$-, and SF$_3$-sketches.

\noindent\textbf{Deletion:}
To delete an item from the SF$_4$-sketch, we first delete it from the Fat-subsketch and then from the Slim-subsketch.
To delete the item $e$ from the Fat-subsketch, we first calculate the $d$ hash functions $h_1(e), h_2(e), \ldots, h_d(e)$ and another $d$ hash functions $f_1(e), f_2(e), \ldots, f_d(e)$ and decrement the $d$ counters $B_1[h_1(e)][f_1(e)]$, $B_2[h_2(e)][f_2(e)]$, $\ldots$, $B_d[h_d(e)][f_d(e)]$ by 1.
To delete the item $e$ from the Slim-subsketch, we leverage the fact stated earlier that before deleting the item from the Fat-subsketch, \textit{the value of a counter in the Slim-subsketch is always less than or equal to the sum of values of all counters in the corresponding bucket in the Fat-subsketch}.
%
To delete the item $e$ from the Slim-subsketch, after deleting it from the Fat-subsketch, for each $i\in[1, d]$, we compare $A_i[h_i(e)]$ with $\sum_{k=1}^{z}B_i[h_i(e)][k]$ and decrement counter $A_i[h_i(e)]$ by 1 if $A_i[h_i(e)] > \sum_{k=1}^{z}B_i[h_i(e)][k]$.
Therefore, one deletion from the Fat-subsketch only needs $d\times z\times b/W$ memory accesses, where $b$ is the number of bits of each counter, $W$ is the size of the \texttt{machine word}, and $b<W$.

\noindent\textbf{Advantages and Limitations: }
The principles behind the SF$_4$-sketch and the SF$_3$-sketch are essentially the same.
The advantage SF$_4$-sketch has over SF$_3$-sketch is that all counters in the Fat-subsketch corresponding to a counter in the Slim-subsketch are now located in the same bucket.
Thus, adding the values of the $z$ counters usually only takes a single memory access.
Based on SF$_4$-sketch, our final version SF$_\text{F}$-sketch aims to minimize the over-estimation error.

\presub
\subsection{SF$_\text{F}$: Reducing Over-estimation
Error (\textbf{The Final Version})}
\label{sec:final:sketch}
The structure of the SF$_\text{F}$-sketch is exactly the same as SF$_4$-sketch.
%
%
%
The key idea behind the SF$_\text{F}$-sketch is that in updating the counters in the Slim-subsketch, we keep the value of each counter in the Slim-subsketch always \emph{\textbf{less than or equal to}} the value of the \textit{\textbf{largest} counter in the corresponding bucket} of the Fat-subsketch during insertion and deletion operations.
%
%
Next, we describe how insertion, deletion, and query operations work in SF$_\text{F}$-sketch followed by an analysis of its error and accuracy.

\noindent\textbf{Insertion: }
The insertion operation of  the SF$_\text{F}$-sketch is exactly the same as the insertion operation of the SF$_4$-sketch.
%

\noindent\textbf{Query: }
The query operation of the SF$_\text{F}$-sketch is exactly the same as the previous versions of the SF-sketch.

\noindent\textbf{Deletion:}
To delete an item from the SF$_\text{F}$-sketch, we first delete it from the Fat-subsketch and then from the Slim-subsketch.
%
%
To delete an item $e$ from the Slim-subsketch, we first check the $d$ buckets $B_1[h_1(e)], B_2[h_2(e)]...B_d[h_d(e)]$.
For each $i\in[1, d]$, if $\max_{k=1}^{z}B_i[h_i(e)][k]$ changes when deleting item $e$ from the Fat-subsketch, we set $A_i[h_i(e)] = \max_{k=1}^{z}B_i[h_i(e)][k]$ if $A_i[h_i(e)]>\max_{k=1}^{z}B_i[h_i(e)][k]$.
Otherwise, we leave the value of $A_i[h_i(e)]$ unchanged.
The key difference between the deletion operation of SF$_\text{F}$-sketch and SF$_4$-sketch is that in SF$_\text{F}$-sketch, we compare the value of $A_i[h_i(e)]$ with $\max_{k=1}^{z}B_i[h_i(e)][k]$ instead of $\sum_{k=1}^{z}B_i[h_i(e)][k]$, which results in significantly reducing the values of counters in the Slim-subsketch.
%
%

\noindent\textbf{Advantages:}
The key advantage of SF$_\text{F}$-sketch over SF$_4$-sketch is that during deletion operation, it significantly reduces the counter values in the Slim-subsketch because in SF$_\text{F}$-sketch, we compare the values of counters in the Slim-subsketch with the values of the largest counters in the corresponding buckets of the Fat-subsketch instead of comparing them with the sum of the values of all counters in the corresponding buckets of the Fat-subsketch.
This significantly reduces the over-estimation error of SF$_\text{F}$-sketch.
Note that SF$_\text{F}$-sketch does not suffer from under-estimation error.
\vspace{0.1in}
\noindent\textbf{Bound on Over-estimation Error:}
As a query is entirely answered from the Slim-subsketch, the over-estimation error of SF$_\text{F}$-sketch is actually the over-estimation error of the Slim-subsketch.
Therefore, next, we calculate the over-estimation error of the Slim-subsketch of the SF$_\text{F}$-sketch.
Let $\alpha$ represent the average number of counters in any given array of the Slim-subsketch that are incremented per insertion.
Note that for the standard CM-sketch, the value of $\alpha$ is equal to 1 because in the
standard CM-sketch, exactly one counter is incremented in each array when inserting an item.
For the Slim-subsketch in the SF$_\text{F}$-sketch, $\alpha$ is less than or equal to 1 because the Fat-subsketch helps in reducing the number of counters that are incremented in the Slim-subsketch per insertion.
For any given item $e$, let $f_{(e)}$ represent its actual frequency and let $\hat{f}_{(e)}$ represent the estimate of its frequency returned by the Slim-subsketch of the SF$_\text{F}$-sketch.
Let $N$ represent the total number of insertions of all items into the SF$_\text{F}$-sketch.
Let $h_i(.)$ represent the hash function associated with the $i^{\text{th}}$ array of the Slim-subsketch, where $1\leqslant i\leqslant d$.
Let $X_{i, (e)}[j]$ be the random variable that represents the difference between the actual frequency $f_{(e)}$ of the item $e$ and the value of the $j^\text{th}$ counter in the $i^{\text{th}}$ array, \ie, $X_{i, (e)}[j] = A_i[j]- f_{(e)}$, where $j=h_i(e)$.
Due to hash collisions, multiple items will be mapped by the hash function $h_i(.)$ to the counter $j$, which increases the value of $A_i[j]$ beyond $f_e$ and results in over-estimation error.
As all hash function have uniformly distributed output, $Pr[h_i(e_1) = h_i(e_2)] = 1/w$.
Therefore, the expected value of any counter $A_i[j]$, where $1\leqslant i\leqslant d$ and $1\leqslant j\leqslant w$,
is $\alpha N/w$.
%
Let $\epsilon$ and $\delta$ be two numbers that are related to $d$ and $w$ as follows:
$d = \lceil\ln(1/\delta)\rceil$ and $w = \lceil \exp/\epsilon \rceil$.
The expected value of $X_{i, (e)}[j]$ is given by the following expression.

\begin{equation}
\begin{aligned}
E(X_{i, (e)}[j])&=E(A_i[j]- f_{(e)})\\
&\leqslant E(A_i[j]) \\
&=\frac{\alpha N}{w}
\leqslant \dfrac{\epsilon\alpha}{\exp}N
\end{aligned}
\label{equ:expectedValX}
\end{equation}

\vspace{0.3in}
Finally, we derive the probabilistic bound on the over-estimation error of the Silm-subsketch of the SF$_\text{F}$-sketch.

\begin{IEEEeqnarray*}{rCl}
Pr[\hat{f_{(e)}} \geqslant f_{(e)} + \epsilon\alpha N] &=& Pr[\forall i, A_i[j] \geqslant f_{(e)} + \epsilon\alpha N]\\
&=& (Pr[A_i[j] - f_{(e)} \geqslant \epsilon\alpha N])^{d} \\
&=& (Pr[X_{i, (e)}[j]\geqslant \epsilon\alpha N]) ^ {d} \\
\IEEEeqnarraymulticol{3}{l}{\text{Substituting the value of $\epsilon\alpha N$ from Equation \eqref{equ:expectedValX} into the}}\\
\IEEEeqnarraymulticol{3}{l}{\text{right side of the equation above, we get}}\\
Pr[\hat{f_{(e)}} \geqslant f_{(e)} + \epsilon\alpha N]&\leqslant& (Pr[X_{i, (e)}[j]\geqslant \exp E(X_{i, (e)}[j]))^{d}\\
\IEEEeqnarraymulticol{3}{l}{\text{Applying Markov's Inequaltiy, we get}}\\
Pr[\hat{f_{(e)}} \geqslant f_{(e)} + \epsilon\alpha N]&\leqslant& \exp ^ {-d}\\
&\leqslant& \delta
\end{IEEEeqnarray*}

\Comment{
\begin{equation*}
\begin{aligned}
Pr[\hat{f_{(e)}} \geqslant f_{(e)} + \epsilon\alpha N]
&= Pr[\forall i, A_i[j] \geqslant f_{(e)} + \epsilon\alpha N]\\
&= (Pr[A_i[j] - f_{(e)} \geqslant \epsilon\alpha N])^{d} \\
&= (Pr[X_{i, (e)}[j]\geqslant \epsilon\alpha N]) ^ {d} \\
&\leqslant (Pr[X_{i, (e)}[j]\geqslant \exp E(X_{i, (e)}[j]))^{d}\\
&\leqslant \exp ^ {-d}\\
&\leqslant \delta
\end{aligned}
\end{equation*}
}


\noindent\textbf{Derivation of Correct Rate:}
The \emph{Correct Rate} of a sketch is defined as the expected percentage of items in the given multi-set for which the query response of the sketch contains no error.

%
%
In deriving the correct rate of SF$_\text{F}$-sketch, we make two assumptions: 1) all hash functions are independent; 2) the Fat-subsketch is large enough to have negligible error.
Before deriving the correct rate, we first prove the following theorem.

\begin{thm}
In the Slim-subsketch, the value of any given counter is equal to the frequency of the most frequent item that maps to it.
\end{thm}

\begin{proof}
We prove this theorem using mathematical induction on number of insertions, represented by $k$.

\noindent\textbf{Base Case, $k$ = 0:}
The theorem clearly holds for the base case because with no insertions, the frequency of the most frequent item is currently 0, which is also the value of all counters.

\noindent\textbf{Induction Hypothesis, $k$ = $n$:}
Suppose the statement of the theorem holds true after $n$ insertions.

\noindent\textbf{Induction Step, $k$ = $n+1$:}
Let $n+1^{\text{st}}$ insertion be of any item $e$ that has previously been inserted $a$ times.
Let $\alpha_i(k)$ represent the values of the counter $A_i[h_i(e)]$ after $k$ insertions, where $0\leqslant i\leqslant d-1$.
There are two cases to consider: 1) $e$ was the most frequent item when $k=n$; 2) $e$ was not the most frequent item when $k=n$.

\noindent\textbf{\textit{Case 1:}}
If $e$ was the most frequent item when $k=n$, then according to our induction hypotheses, $\alpha_i(n)=a$.
After inserting $e$, it will still be the most frequent item and its frequency increases to $a+1$.
The counter $A_i[h_i(e)]$ will be incremented once.
Consequently, we get $\alpha_i(n+1)=a+1$.
Thus for this case, the theorem statement holds because the value of the counter $A_i[h_i(e)]$ after insertion is still equal to the frequency of the most frequent item, which is $e$.

\noindent\textbf{\textit{Case 2:}}
If $e$ was not the most frequent item when $k=n$, then according to our induction hypotheses, $\alpha_i(n)>a$.
After inserting $e$, it may or may not become the most frequent item.
If it becomes the most frequent item, it means that $\alpha_i(n)=a+1$ and as our SF$_\text{F}$ scheme, the counter $A_i[h_i(e)]$ will stay unchanged.
Consequently, we get $\alpha_i(n+1)=\alpha_i(n)=a+1$.
Thus for this case, the theorem statement again holds because the value of the counter $A_i[h_i(e)]$ after insertion is equal to the frequency of the new most frequent item, which is $e$.

After inserting $e$, if it does not become the most frequent item, then it means $\alpha_i(n)>a+1$ and as our SF$_\text{F}$-sketch scheme, the counter $A_i[h_i(e)]$ will stay unchanged.
%
Consequently, $\alpha_i(n+1)=\alpha_i(n)>a+1$.
Thus, the theorem again holds because the value of the counter $A_i[h_i(e)]$ after insertion is still equal to the frequency of the item that was the most frequent after $n$ insertions.
\end{proof}

\vspace{0.09in}
Next, we derive the correct rate of the SF$_\text{F}$-sketch.
%
Let $v$ be the number of distinct items inserted into the slim-subsketch and are represented by $e_1, e_2, \ldots, e_v$.
Without loss of generality, let the item $e_{l+1}$ be more frequent than $e_l$, where $1\leqslant l\leqslant v-1$.
Let $X$ be the random variable representing the number of items hashing into the counter $A_i[h_i(e_l)]$ given the item $e_l$, where $0\leqslant i \leqslant d-1$ and $1\leqslant l\leqslant v$.
Clearly, $X\sim\text{Binomial}(v-1, 1/w)$.

%
From Theorem 1, we conclude that if $e_l$ has the highest frequency among all items that map to the given counter $A_i[h_i(e_l)]$, then the query result for $e_l$ will contain no error.
Let $\mathbb{E}$ be the event that $e_l$ has the maximum frequency among $x$ items that map to $A_i[h_i(e_l)]$.
The probability $P\{\mathbb{E}\}$ is given by the following equation:

\vspace{-0.15in}
\begin{equation*}
P\{\mathbb{E}\}=\begin{pmatrix}{l-1}\\{x-1}
\end{pmatrix}/
\begin{pmatrix}{v-1}\\{x-1}
\end{pmatrix}
~~~~~~~~~(\text{where } x\leqslant l)
\vspace{-0.05in}
\end{equation*}
Let $P'$ represent the probability that the query result for $e_l$ from any given counter contains no error.
It is given by:

\noindent
\begin{equation*}
\begin{aligned}
P'&=\sum_{x=1}^lP\{\mathbb{E}\}\times P\{X=x\}\\
&=\sum_{x=1}^l \dfrac{\binom{l-1}{x-1}}{\binom{v-1}{x-1}}\binom{v-1}{x-1}\Big(\frac{1}{w}\Big)^{x-1}\Big(1-\frac{1}{w}\Big)^{v-x}
=\Big(1-\frac{1}{w}\Big)^{v-l}
\end{aligned}
\vspace{-0.05in}
\end{equation*}
As there are $d$ counters, the overall probability that the query result of $e_l$ is correct is given by the following equation.
\vspace{-0.05in}
\begin{equation*}
P_{\text{CR}}\{e_l\}=1-\bigg(1-\Big(1-\frac{1}{w}\Big)^{v-l}\bigg)^d
\vspace{-0.05in}
\end{equation*}
The equality above holds when all $v$ items have different frequencies.
If two or more items have equal frequencies, the correct rate increases slightly.
Consequently, the expected correct rate $Cr$ of slim-subsketch is bound by:
\vspace{-0.05in}
\begin{equation}
Cr\geqslant \dfrac{\sum_{l=1}^v P_{\text{CR}}\{e_l\}}{v}
=\dfrac{\sum_{l=1}^v \left(1- \left( 1-(1-\frac{1}{w})^{v-l}\right)^d   \right)}{v}
\label{equ:cr}
\end{equation} 


\section{Implementation} 
\label{sec:implementation}

In this section, we describe our implementation of the sketches on two different computing platforms namely CPU and GPU.
We extensively tested and evaluated SF-sketch and compared its performance with prior sketches on these two platforms.
Next, we first describe our implementation on the CPU platform and then describe our implementation on the GPU platform.

\presub
\subsection{CPU Implementation} \postsub
Our CPU platform comprised a machine with dual 6-core CPUs (24 threads, Intel Xeon CPU E5-2620 @2 GHz) and 62 GB total system memory.
Each CPU has three levels of cache memory: L1, L2, and L3.
L1 cache is comprised of two 32KB caches, where one cache acts as the data cache and the other acts as the instruction cache.
L2 cache is a single 256KB cache and L3 cache is a single 15MB cache.
To evaluate the schemes in different types of settings, our implementations on the CPU platform include both single-thread implementation as well as multi-thread implementation.
We used C++ as the programming language.
In single-thread implementation, for each sketch, we implemented the entire insertion, deletion, and query process within a single thread.
In multi-thread implementation, we run each query in a dedicated thread and process it completely inside that thread, observing near-linear growth in query speed with the increase in the number of threads.
%
We will present the results on query speed in more detail in Section \ref{subsubsec:OnMulti-coreCPUPlatform}.

\presub
\subsection{GPU Implementation} \postsub
As GPUs have seen wide acceptance for high-speed data processing, we implemented our sketches on GPUs as well.
For these implementations, we employ the basic architecture of GAMT~\cite{GPU-GAMT}.
More specifically, we evaluated the sketches on GPU platform using CUDA 5.0 architecture.
We performed our experiments on a NVIDIA GPU (Tesla C2075, 1147 MHz, 5376 MB device memory, 448 CUDA cores).
We implemented our sketches on GPU using two prevalent techniques: batch processing and multi-stream pipelining.
Next, we describe our implementations for these two techniques.

\subsubsection{Batch Processing} 
\label{sect:gpu:batch}
Our system architecture is based on CUDA~\cite{CUDA-Best}, the well-known parallel computing platform created by NVIDIA.
In our implementation, a typical query cycle proceeds in following three steps: (1) copy the incoming queries from the CPU to the GPU, (2) execute the query kernel, and (3) copy the result from the GPU back to the CPU.
A kernel in CUDA is a function that is called on CPU but executed on GPU.
A query kernel is configured with a series of thread blocks, where each block is comprised of a group of working threads.
As GPU chips have hundreds and even thousands of cores, batch processing is needed to accelerate GPU-based implementations.
Each batch is first filled with a group of independent queries, and then transferred to and executed on the GPU, \ie, as soon as a query arrives, it is buffered until there are enough queries to fill the current batch of queries before transferring the batch to GPU for processing by the query kernels.
Note that in practice, not all the queries are processed simultaneously, but rather GPU's scheduler decides when to process which query.
As CPUs support less parallelism compared to GPUs and the additional memory accesses to the CPUs may deteriorate the batch processing performance of GPUs, in our implementation, all $d$ arrays are stored on the GPU to ensure that the operations to access the arrays are executed completely within the GPU.
%

\subsubsection{Multi-Stream Pipeline} 
\label{sect:gpu:stream}
As discussed earlier, batch processing is required to take maximum advantage of the massive parallelization that GPU enables.
However, waiting for enough queries to fill a batch before sending the batch to GPU results in unnecessary delays.
Furthermore, while a large batch does boost the throughput of the GPU, it increases the waiting time before a batch fills and is transferred to GPU for processing.
This means that the query that arrived at the start of the current batch will experience significant latency before it is processed.
To resolve this throughput-latency dilemma, we utilize the multi-stream technique featured in NVIDIA Fermi GPU architecture~\cite{GPU-NDN, GPU-GAMT}.
A stream, in this context, is a sequence of operations that must be executed in a certain order.

As per CUDA architecture, data transfers and kernel executions within different streams can be concurrent as long as the device supports concurrent operations and the host memories used to exchange data between the CPU and the GPU are page-locked.
In this way, when one stream is copying data between the CPU and GPU, another stream can execute query kernels in parallel.
As a result, the streams behave as a multi-stage pipeline and reduce the total processing time.

Furthermore, a large batch can be divided into several smaller ones, reducing the average lookup latency while keeping the throughput high.
Given a batch of requests and a sequence of active streams, the task mapping should be performed in as balanced way as possible to efficiently use GPU's parallelism.
Let $b$ and $n$ denote the batch size and the number of active streams, respectively.
If $b$ is just a multiple of $n$, the whole batch can be evenly divided.
Otherwise, the first $b~\%~n$ streams may need to perform an extra operation.
In our implementation, after dividing the whole batch into multiple smaller batches of approximately identical sizes, we use their offset and size information for task mapping.
Each small batch is processed by the specified stream, and all streams are launched one after the other to work as a multi-stage pipeline.


	\presec
\section{Experimental Results}
\postsec
\label{sec:experiments}

\begin{figure*}[t]
\centering
	\begin{minipage}[t]{0.241\textwidth}
	\includegraphics[width=1\textwidth]{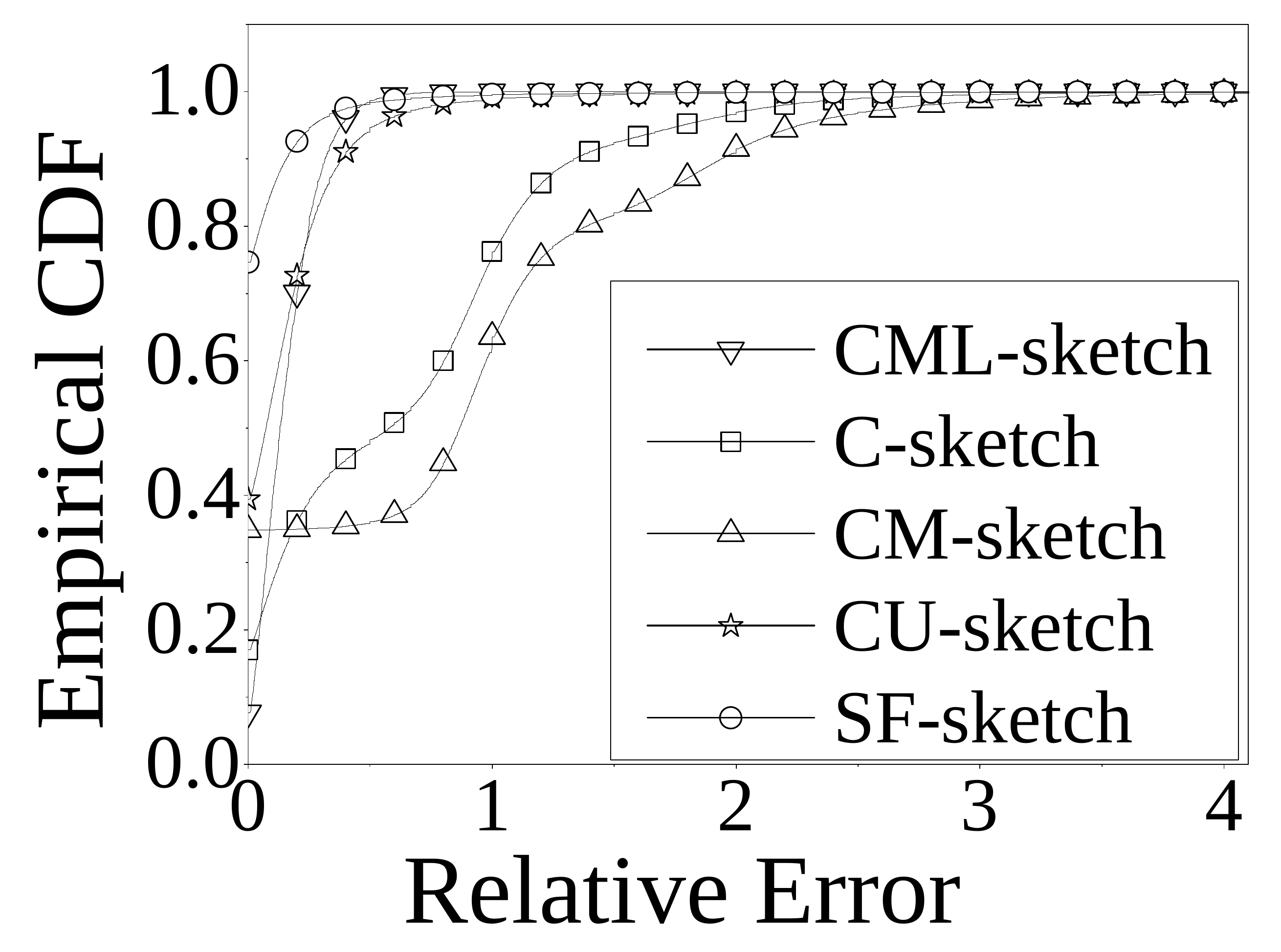}
	\caption{CDF of relative error (uniform).}
	\label{fig:cpu:exp:CDF-RE-UNF}
	\end{minipage}
	\begin{minipage}[t]{0.239\textwidth}
	\includegraphics[width=1\textwidth]{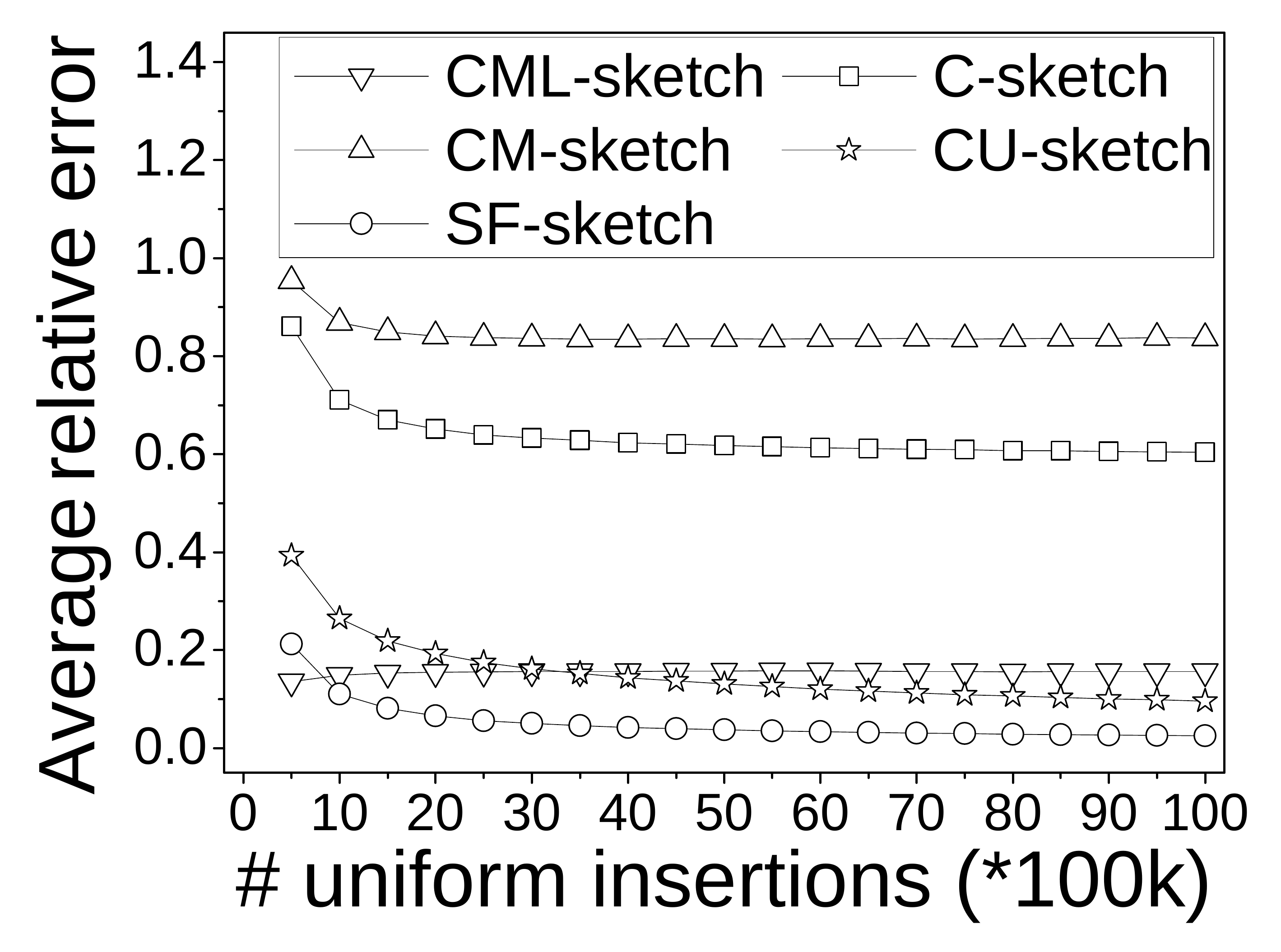}
	\caption{Average relative error vs. number of insertions (uniform).}
	\label{fig:cpu:exp:ARE-UNF-INS}
	\end{minipage}
	\begin{minipage}[t]{0.239\textwidth}
	\includegraphics[width=1\textwidth]{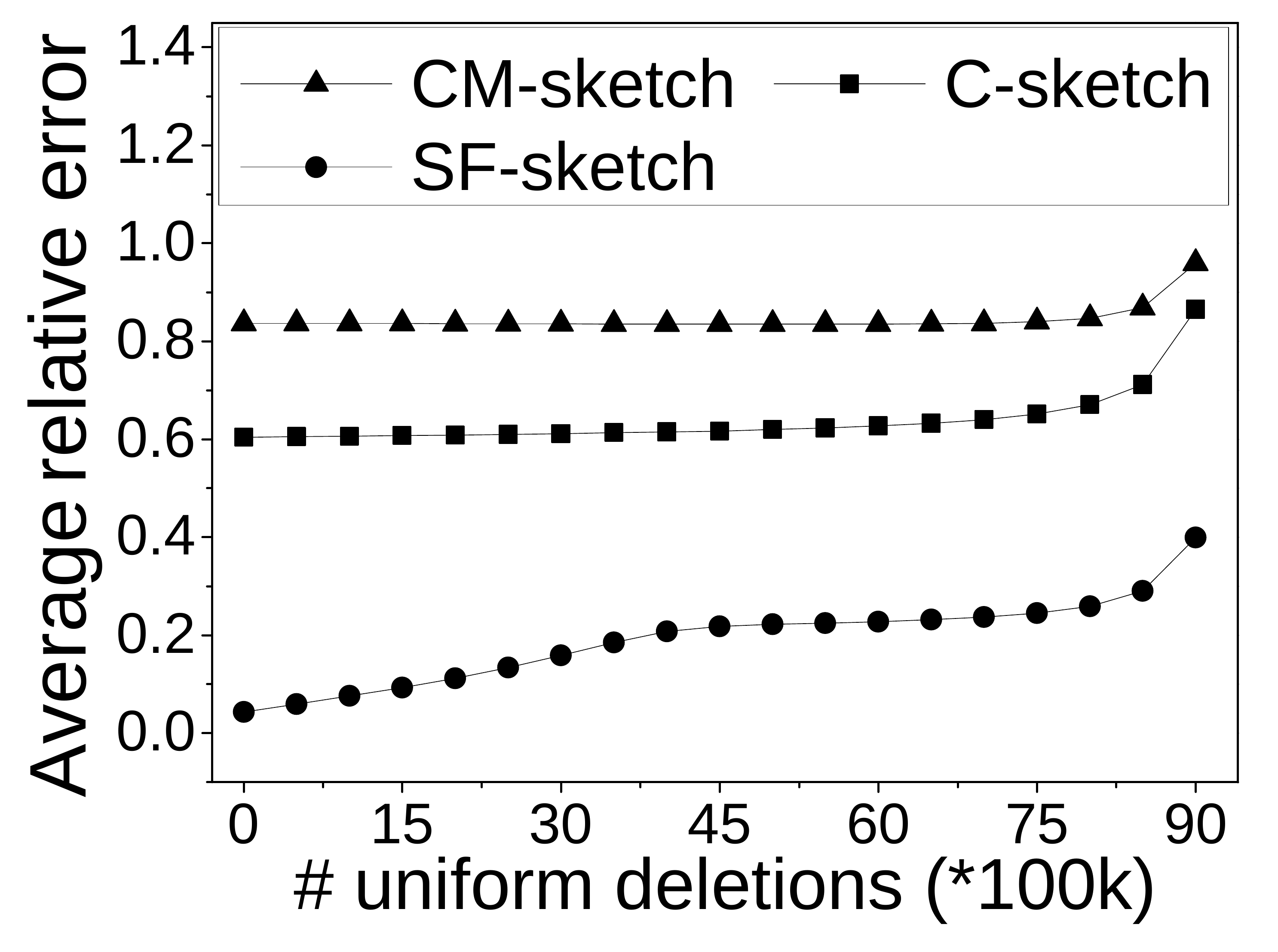}
	\caption{Average relative error vs. number of deletions (uniform).}
	\label{fig:cpu:exp:ARE-UNF-DEL}
	\end{minipage}
	\begin{minipage}[t]{0.242\textwidth}
	\includegraphics[width=1\textwidth]{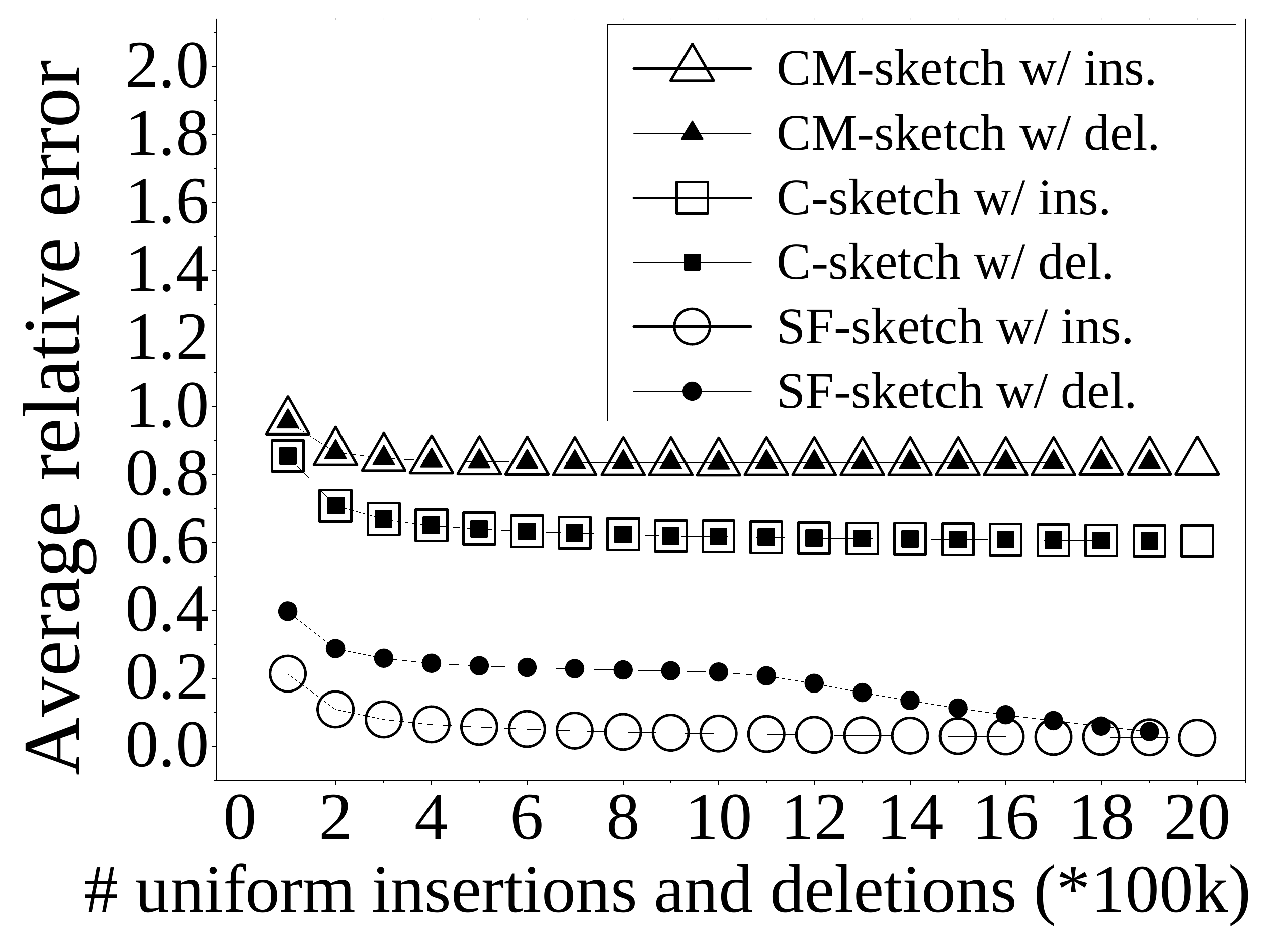}
	\caption{Increase in error due to deletions (uniform).}
	\label{fig:cpu:exp:ARE-UNF-INS-DEL}
	\end{minipage}
%
	\begin{minipage}[t]{0.241\textwidth}
	\includegraphics[width=1\textwidth]{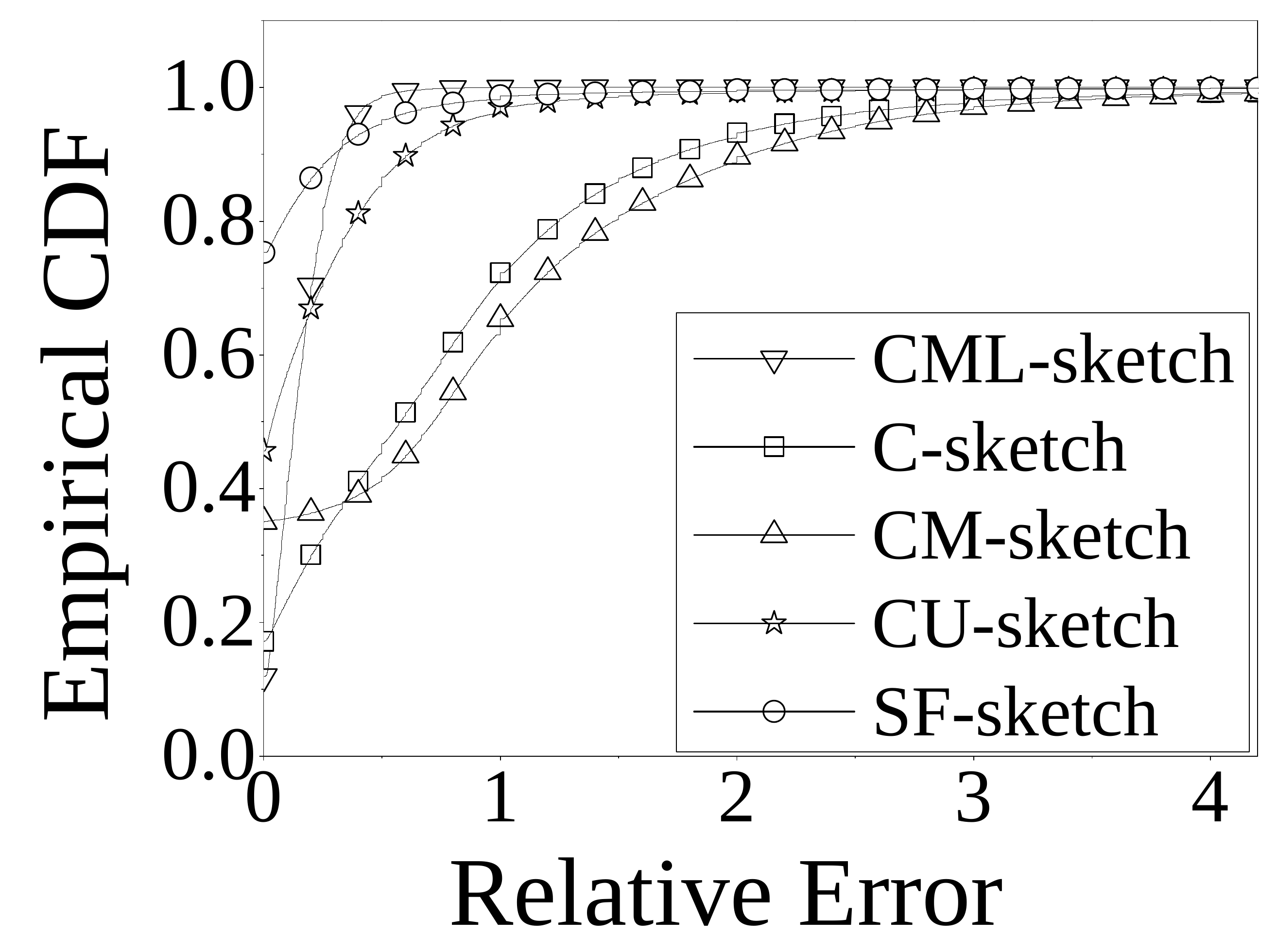}
	\caption{CDF of relative error (skewed).}
	\label{fig:cpu:exp:CDF-RE-ZIPF}
	\end{minipage}
	\begin{minipage}[t]{0.237\textwidth}
	\includegraphics[width=1\textwidth]{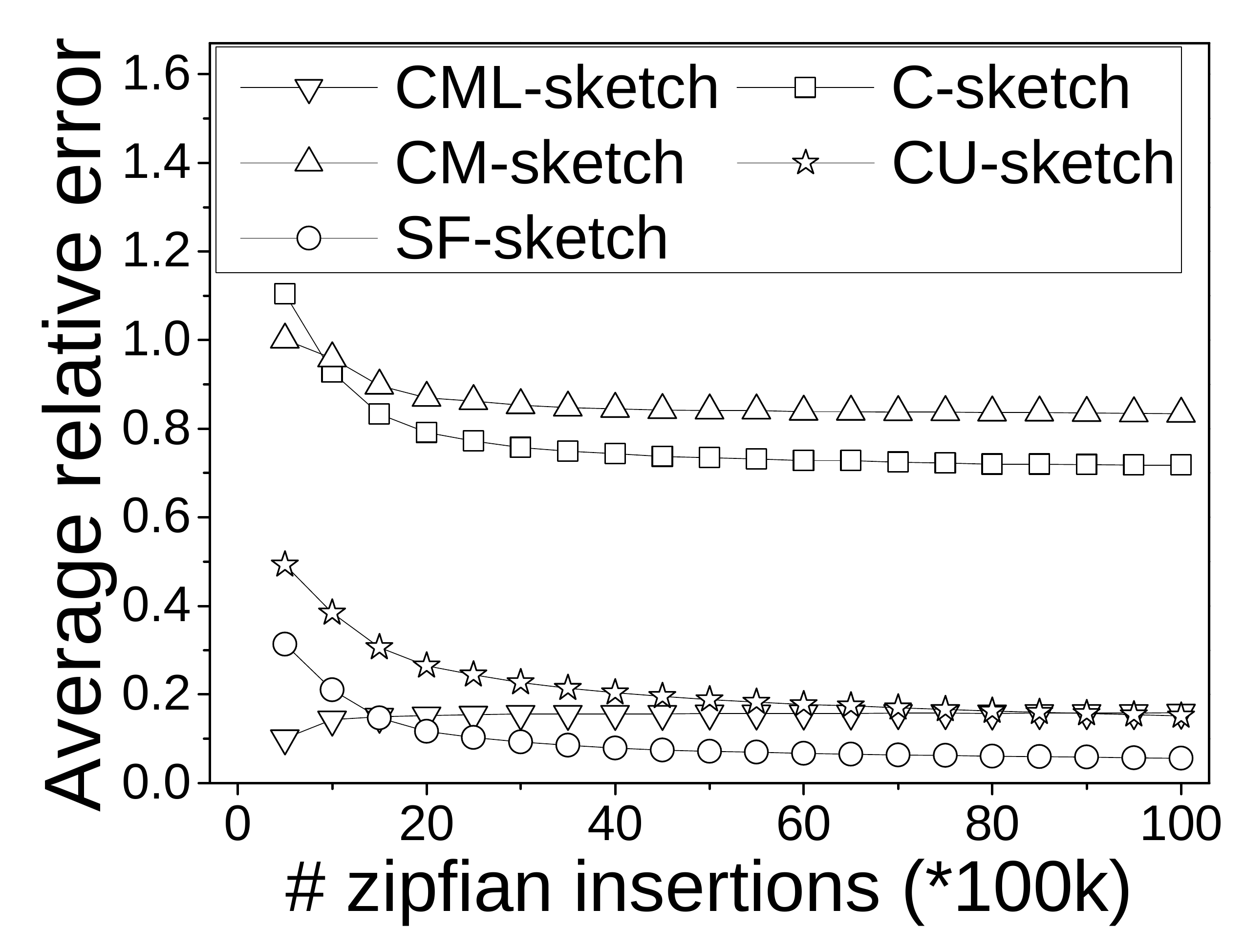}	
	\caption{Average relative error vs. number of insertions (skewed).}
	\label{fig:cpu:exp:ARE-ZIPF-INS}
	\end{minipage}
	\begin{minipage}[t]{0.237\textwidth}
	\includegraphics[width=1\textwidth]{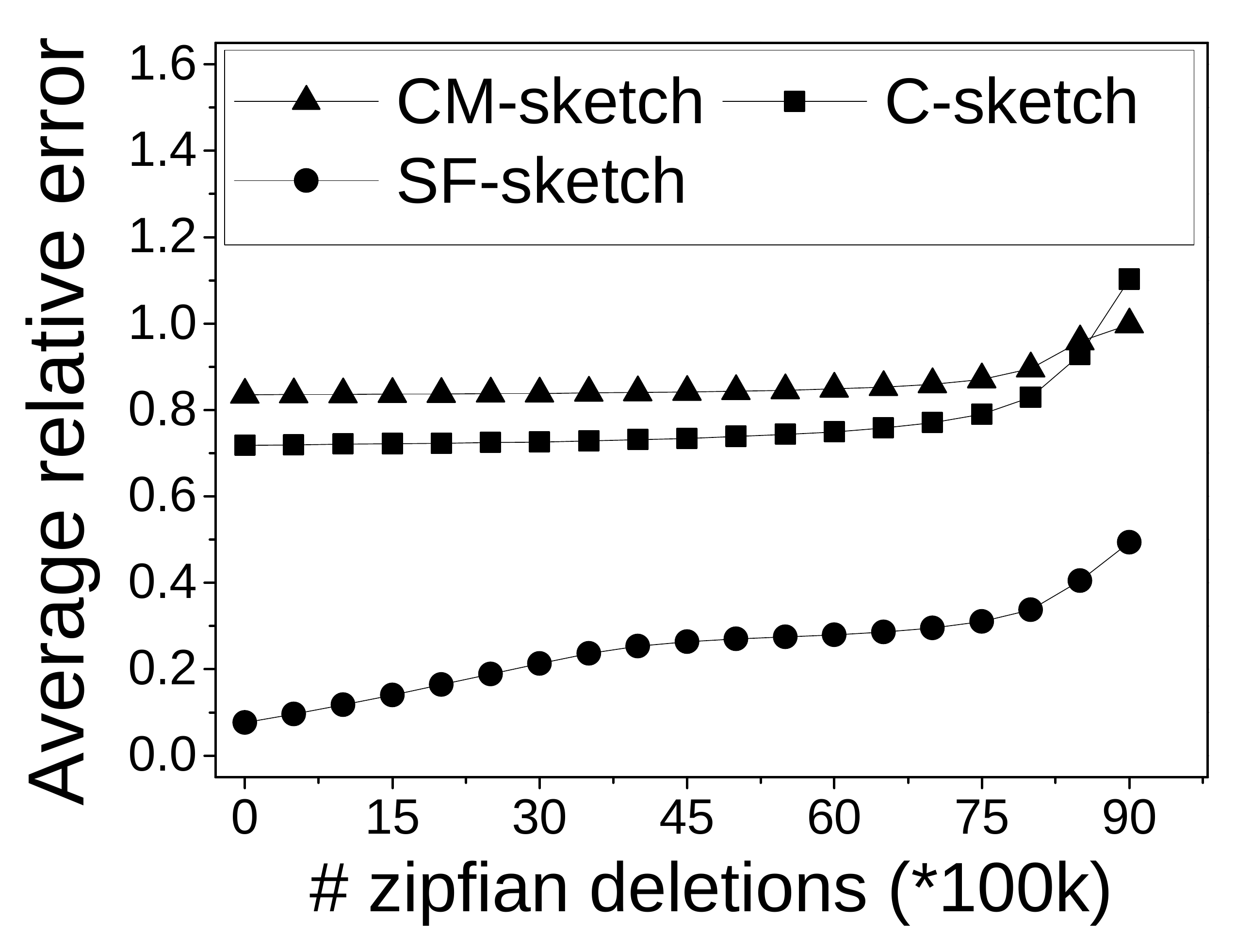}
	\caption{Average relative error vs. number of deletions (skewed).}
	\label{fig:cpu:exp:ARE-ZIPF-DEL}
	\end{minipage}
	\begin{minipage}[t]{0.241\textwidth}
	\includegraphics[width=1\textwidth]{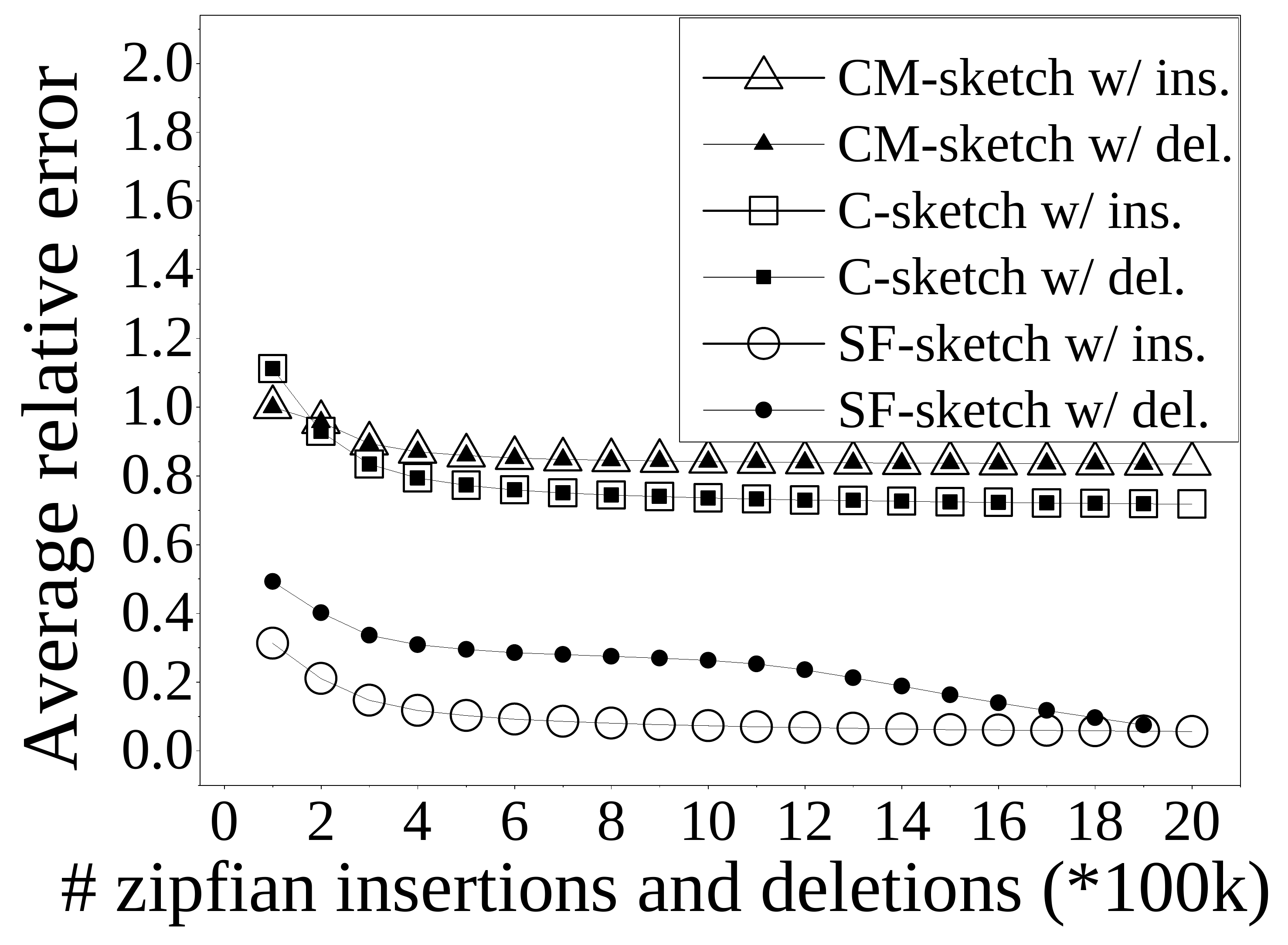}
	\caption{Increase in error due to deletions(skewed).}
	\label{fig:cpu:exp:ARE-ZIPF-INS-DEL}
	\end{minipage}
	\postfig \postfig
\end{figure*}

We conducted extensive experiments to evaluate the performance of our SF$_\text{F}$-sketch in terms of \emph{accuracy} and \emph{speed}.
Onwards, we will refer to the SF$_\text{F}$-sketch as simply the SF-sketch.
For comparison, we also implemented and evaluated the performance of four well known sketches, namely the Count-sketch (C-sketch) \cite{countsketch}, the CM-sketch \cite{cmsketch} and the CU-sketch \cite{cusketch} and one recently proposed sketch, namely the CML-sketch \cite{cmlog}. \textit{CML-sketch and CU-sketch do NOT support deletions. CML-sketch and Count-sketch suffer from both over-estimation and under-estimation errors.}

\subsection{Experimental Setup}
\noindent\textbf{Datasets:} We use three types of datasets: real world traffic, uniform dataset, and skewed dataset.
The real world network traffic trace is captured by the main gateway of our campus, while the uniform and the skewed datasets are generated by the well known YCSB~\cite{YCSB}.
We keep the skewness of our skewed dataset equal to the \texttt{default value} for YCSB, which is 0.99.
We use Memcached~\cite{memcached} to record the real frequency of each item to establish the ground truth.

\noindent\textbf{Experimental Comparison}:
As the memory of the monitor is cheap and large enough, thus we assign the same size of memory for the Slim-subsketch and the state-of-the-art sketches both of which will be transmitted to the collector.
For update experiments, we compare them by varying item frequencies and operation size, \ie, the number of insertion and deletion operations.
%

\presub
\subsection{Experiments on Accuracy}
\postsub
We use \emph{relative error} ($RE$) to quantify the accuracy of sketches.
Let $f_e$ represent the actual frequency of an item $e$ and let $\hat{f}_e$ represent the estimate of the frequency returned by the sketch, the relative error is defined as the ratio $|\hat{f}_e-f_e|/f_e$.
To evaluate accuracy, we used 100K distinct items and fixed parameter setting ($d$ = 5, $w$ = 40000, $z$ = 3).
We calculated relative errors for different sketches in three settings: (1) by incrementally increasing the number of insertion operations; (2) by incrementally increasing the number of deletion operations; and (3) by first increasing the number of insertion operations and then deleting the inserted items one by one in reverse order.
We performed experiments in these three settings for both uniform and skewed workloads.
We also conducted experiments to quantify the effect of system parameters on the performance of the sketches.
In all our experiments on accuracy evaluation, we use 100K ($=100\times 10^3$) distinct items in total.
%


\presub
\subsubsection{Uniform Workload}
\postsub
%
%

\noindent\textbf{Relative Error CDF:}
\textit{Our experimental results show that the percentage of items for which the relative error of our SF-sketch is less than 1\% is 74.51\%, which is 18.8, 4.3, 2.1 and 1.9 times higher than the corresponding percentages for CML, C, CM and CU-sketches, respectively.
}
Figure~\ref{fig:cpu:exp:CDF-RE-UNF} reports the empirical cumulative distribution function (CDF) of relative error for the 100K distinct items after a total of 10M ($=10\times 10^6$) insertions.
Specifically, we first inserted the 100K distinct items for a total of 10M times such that the probability of occurrence for each item was uniformly distributed, and then calculated the relative errors in the estimates of the frequencies of those 100K distinct items.
In this way, we got 100K values of relative error for each of the five sketches (CML, C, CM, CU and SF-sketches).
We then plotted a CDF using the 100K relative error values for each sketch.
We observe from Figure~\ref{fig:cpu:exp:CDF-RE-UNF} that the CDF of the SF-sketch is not only higher than that of the other four sketches but also ascends sharply near relative error of 0.
This indicates that the relative error in the estimate of the frequencies of most items, calculated from the SF-sketch, is very close to 0.
%

%


\noindent\textbf{Relative Error vs. \# of Insertions:}
\textit{Our experimental results show that the average relative error of SF-sketch is [0.6 to 6.2], [4.0 to 24.0], [4.4 to 33.1], and [1.8 to 3.8] times smaller than the average relative errors of CML, C, CM, and CU-sketches, respectively.}
Figure~\ref{fig:cpu:exp:ARE-UNF-INS} plots the average relative errors in the estimate of the frequencies of the 100K distinct items obtained from the five sketches for different number of insertions.
%
We observe from this figure that the average relative errors of the five sketches converge to different fixed values with increasing number of insertions with our SF-sketch being the most accurate.
The converged average relative error of our SF-sketch is 6.2, 24.0, 33.1 and 3.8 times smaller than the converged average relative errors of CML, C, CM, and CU-sketch, respectively.


%

\noindent\textbf{Relative Error vs. \# of Deletions:}
\textit{Our experimental results show that the average relative error of SF-sketch is [2.1 to 23.8] and [2.4 to 33.0] times smaller than the average relative errors of C and CM-sketches, respectively.}
Figure~\ref{fig:cpu:exp:ARE-UNF-DEL} plots the average relative errors in the estimate of the frequencies of the 100K distinct items obtained from the three sketches for different number of deletions.
Before starting the deletions, we inserted the 100K items 10M times for each sketch.
Note that Figure~\ref{fig:cpu:exp:ARE-UNF-DEL} does not include results for CML and CU-sketches because they do not support deletions.



\noindent\textbf{Increase in Error due to Deletions:}
\textit{Our experimental results show that our SF-sketch looses some accuracy due to deletions, while C and CM-sketches do not. Despite that, the average relative error of SF-sketch is still [2.1 to 24.1] times lower compared to the C-sketch and [2.4 to 33.4] times lower compared to the CM-sketch.}
In this experiment, we first inserted the 100K items 10M times, and then deleted them in reverse order as the insertions.
After every 100K insertions, we calculated the average relative error for the 100K distinct items and plotted them in Figure~\ref{fig:cpu:exp:ARE-UNF-INS-DEL}.
Similarly, after every 100K deletions, we calculated the average relative error for the 100K distinct items and plotted them in Figure~\ref{fig:cpu:exp:ARE-UNF-INS-DEL}.
The lines with a hollow square/triangle/circle indicate average relative errors calculated after insertion operations, while the lines with solid square/triangle/circle indicate average relative errors calculated after deletion operations.
Again, Figure~\ref{fig:cpu:exp:ARE-UNF-INS-DEL} does not include results for CML and CU-sketches because they do not support deletions.
We observe from this figure that the lines with a hollow square/triangle for the C and CM-sketches perfectly track the corresponding lines with solid ones, which means that deletion operations do not reduce accuracy in the C and CM-sketches.
However, the line with a hollow circle for the SF-sketch lie below the corresponding line with solid circle showing that the deletion operation deteriorates the accuracy of SF-sketch.
The reason is that when deleting an item, SF-sketch cannot always decrement all counters of Slim-subsketch that it incremented when inserting that item because decrementing all those counters can lead to underestimation error.
Thus, SF-sketch supports deletions at the cost of slightly reduced accuracy after deletions.



\begin{figure*}[t]
	\centering
		\begin{minipage}[t]{0.233\textwidth}
			\includegraphics[width=1\textwidth]{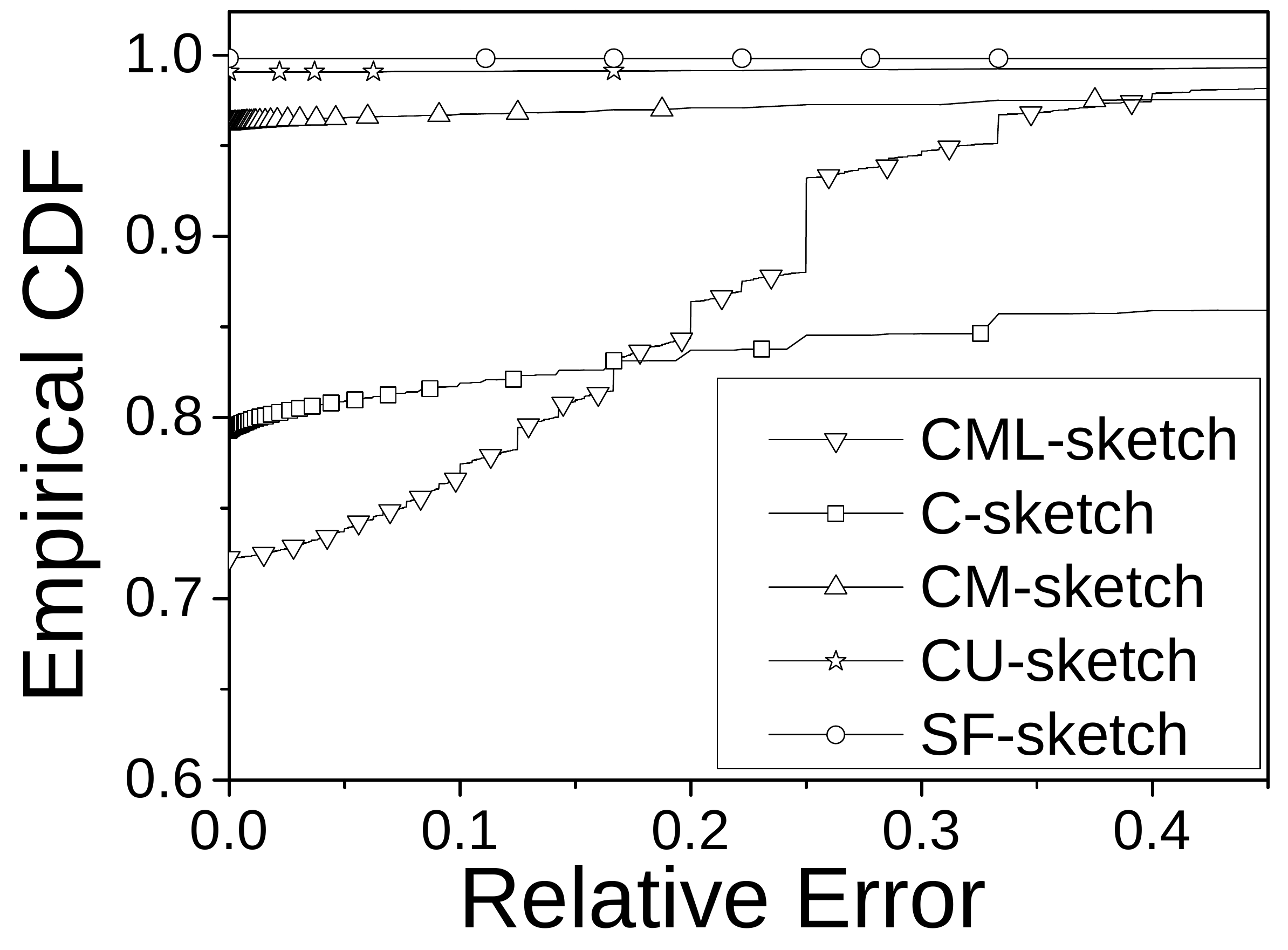}
			\caption{CDF of relative error (real).}
			\label{fig:real:cdf}
		\end{minipage}
			\begin{minipage}[t]{0.23\textwidth}
				\includegraphics[width=1\textwidth]{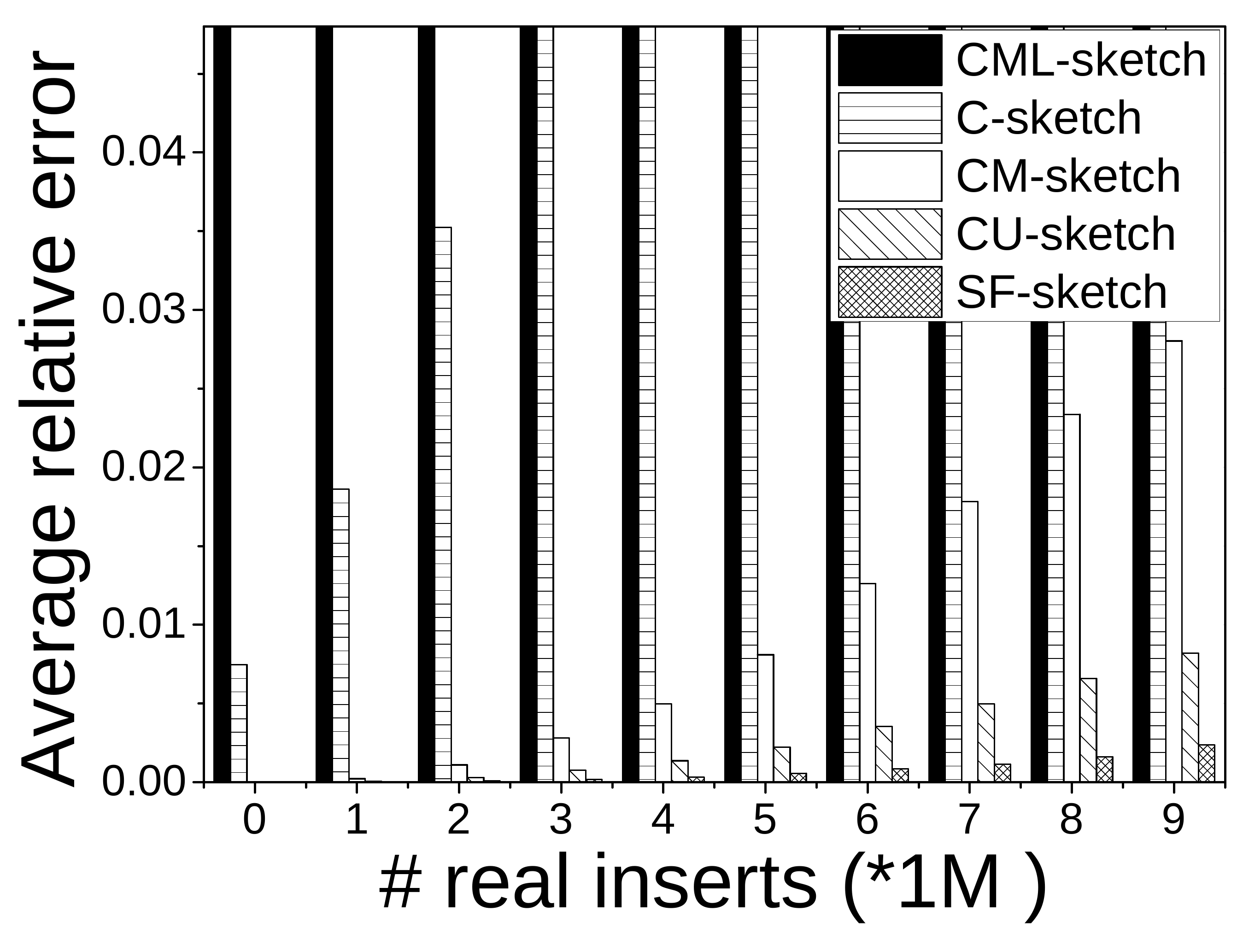}
				\caption{Average relative error vs. number of insertions (real).}
				\label{fig:real:insert}
			\end{minipage}	
	\begin{minipage}[t]{0.223\textwidth}
		\includegraphics[width=1\textwidth]{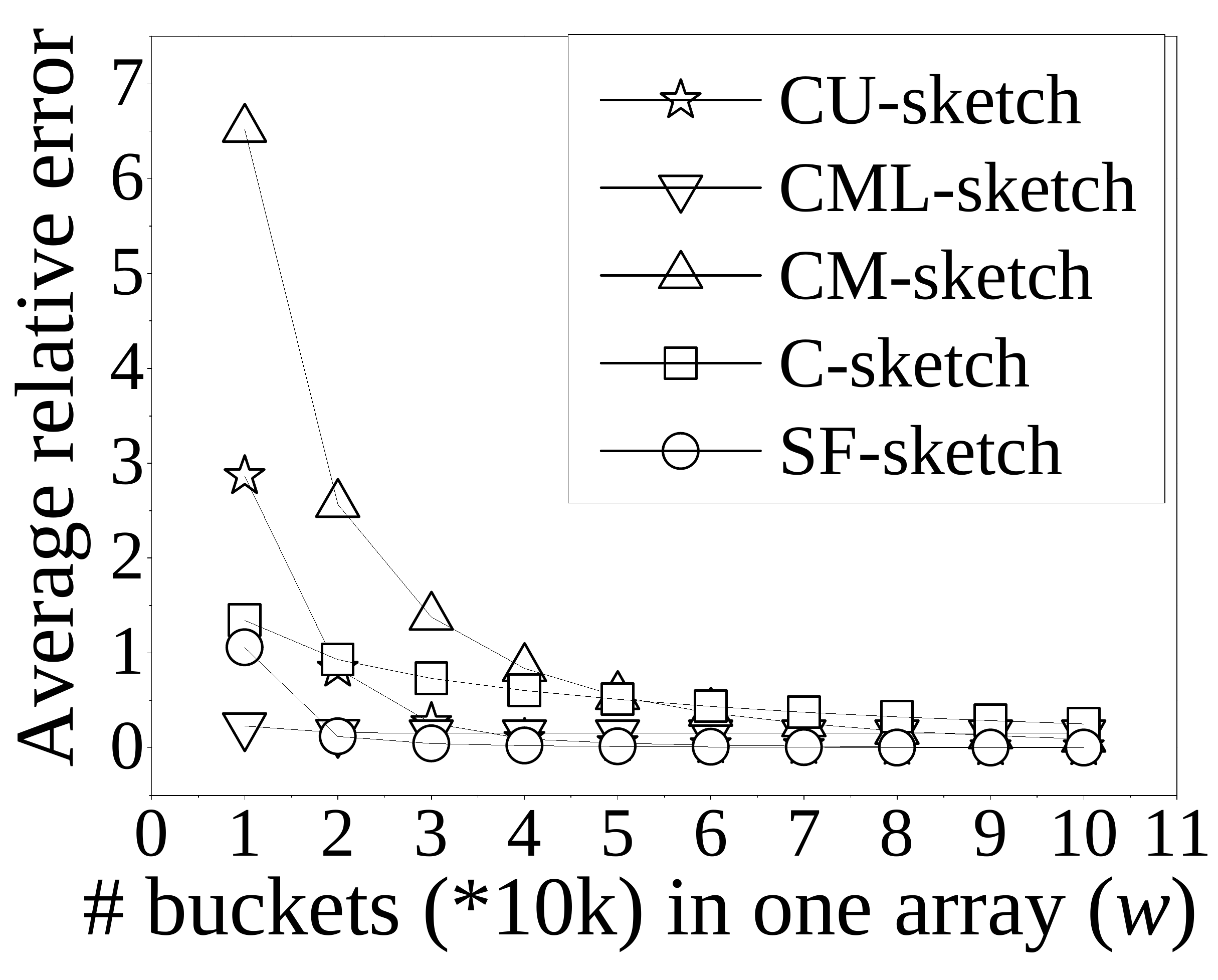}
		\caption{Accuracy vs. \# of buckets per array $w$.}
		\label{fig:cpu:exp:opt-fixed-d}
	\end{minipage}
	\begin{minipage}[t]{0.23\textwidth}
		\includegraphics[width=1\textwidth]{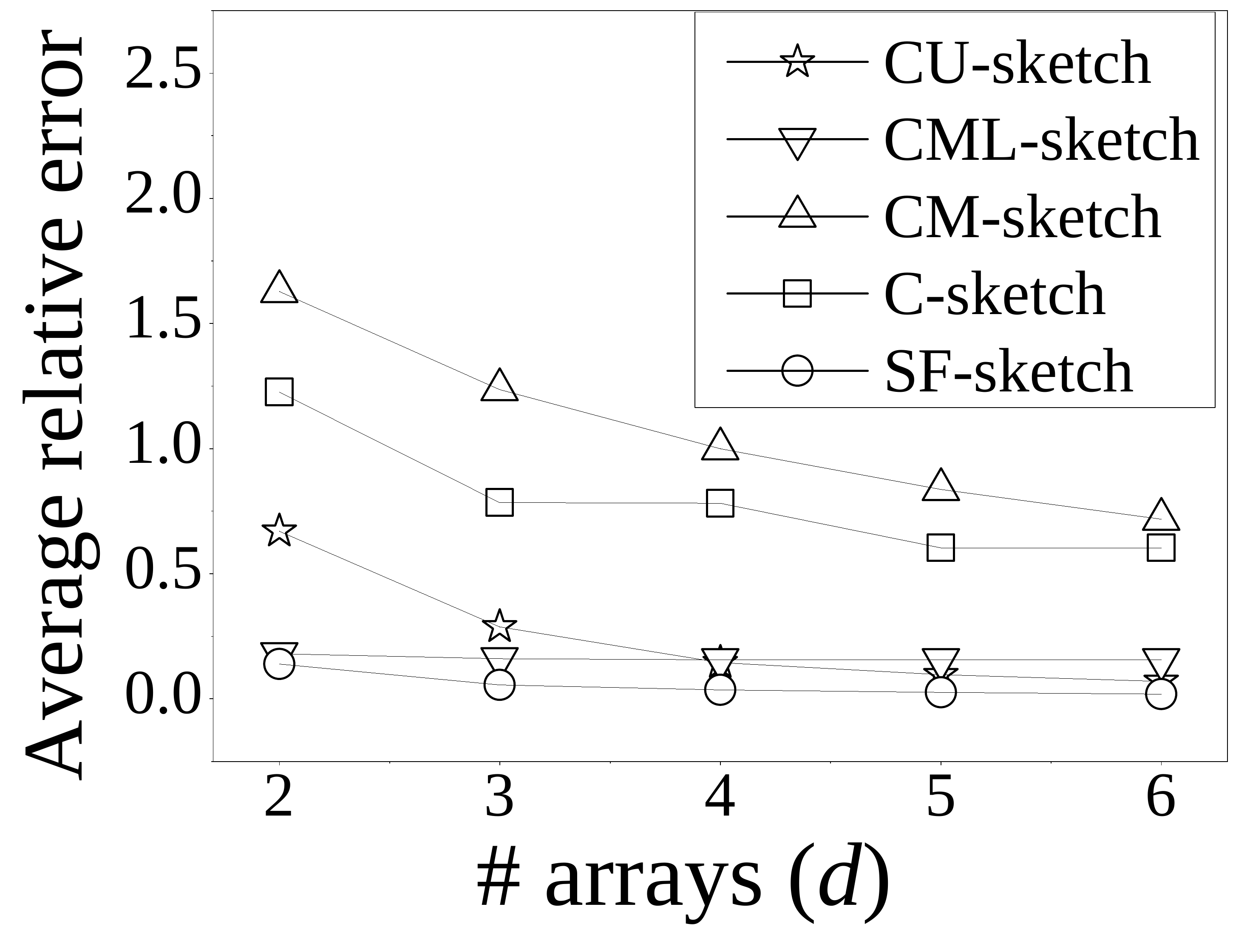}
		\caption{Accuracy vs. \# of arrays $d$.}
		\label{fig:cpu:exp:opt-fixed-w}
	\end{minipage}
	%
	\postfig \postfig
\end{figure*}

\begin{figure*}[t]
	\centering
	\begin{minipage}[t]{0.23\textwidth}
		\includegraphics[width=1\textwidth]{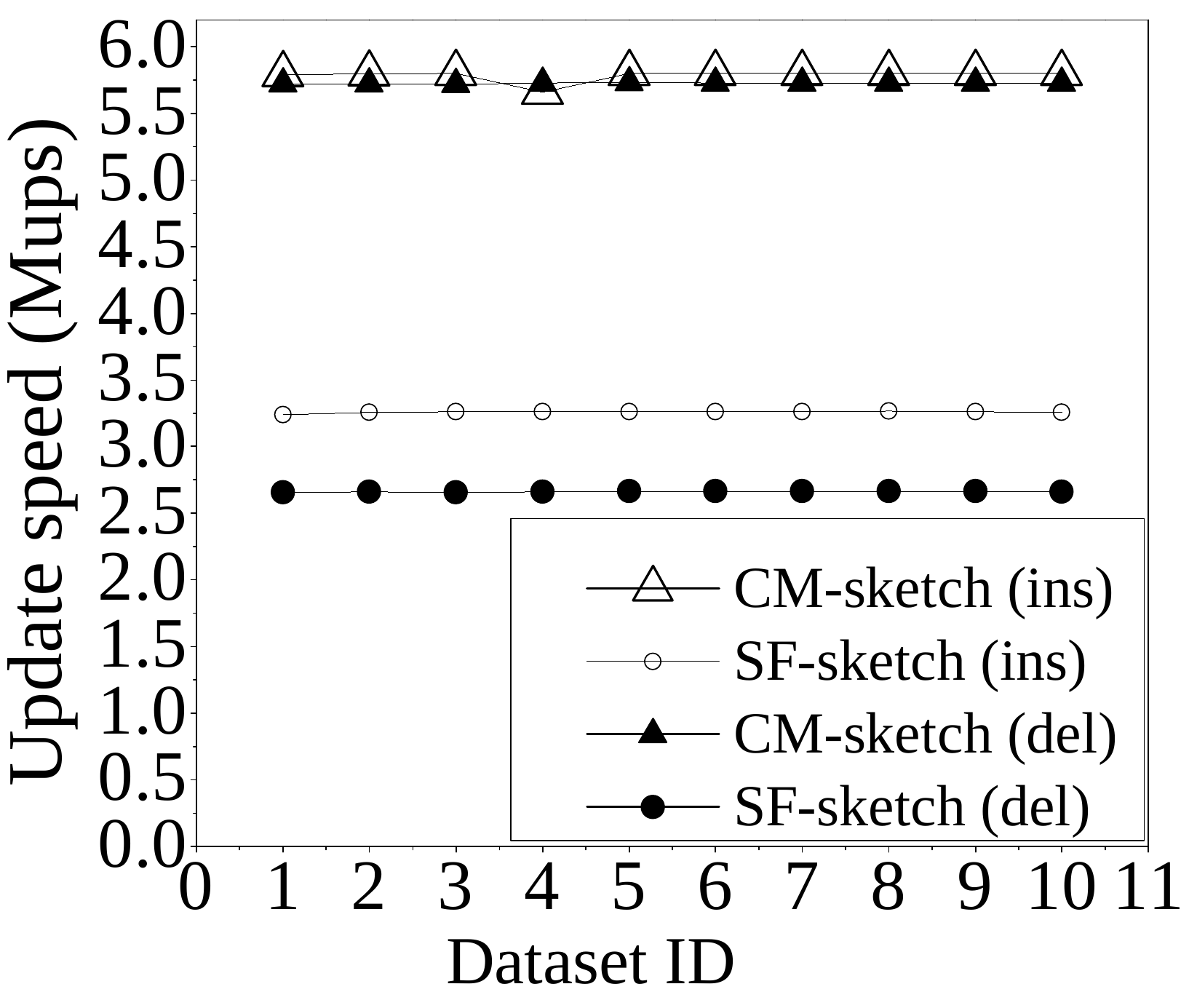}
		\caption{The insertion and deletion speed of CM-sketch and SF-sketch.}
		\label{fig:cpu:exp:updatespeed}
	\end{minipage}
	\centering
	\begin{minipage}[t]{0.23\textwidth}
		\includegraphics[width=1\textwidth]{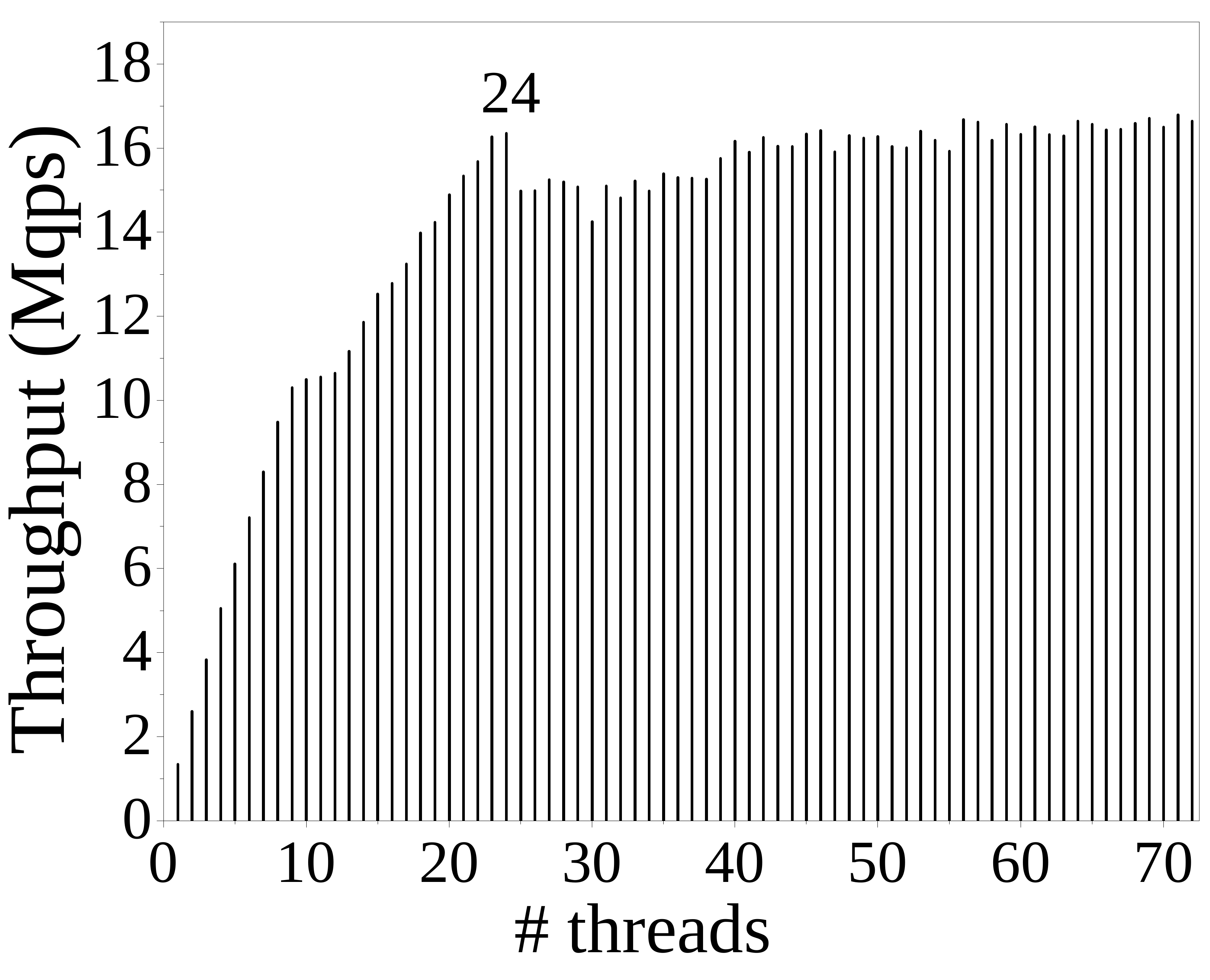}
		\caption{Throughput vs. \# of threads.}
		\label{fig:cpu:exp:multi-threads}
	\end{minipage}
	%
	\begin{minipage}[t]{0.23\textwidth}
		\includegraphics[width=1\textwidth]{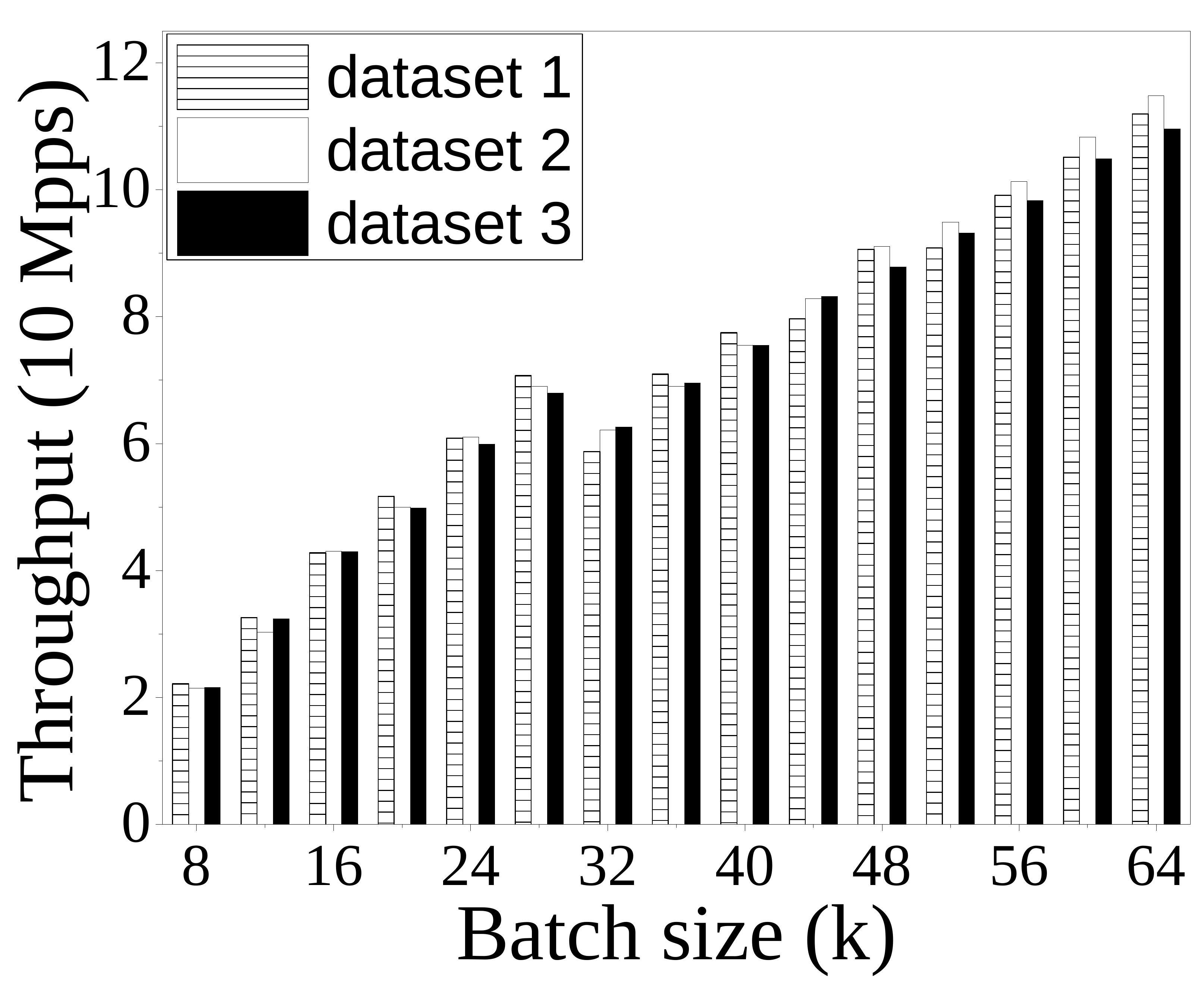}
		\caption{GPU Throughput vs. batch size.}
		\label{fig:gpu:exp:speed}
	\end{minipage}
		%
		\begin{minipage}[t]{0.23\textwidth}
			\includegraphics[width=1\textwidth]{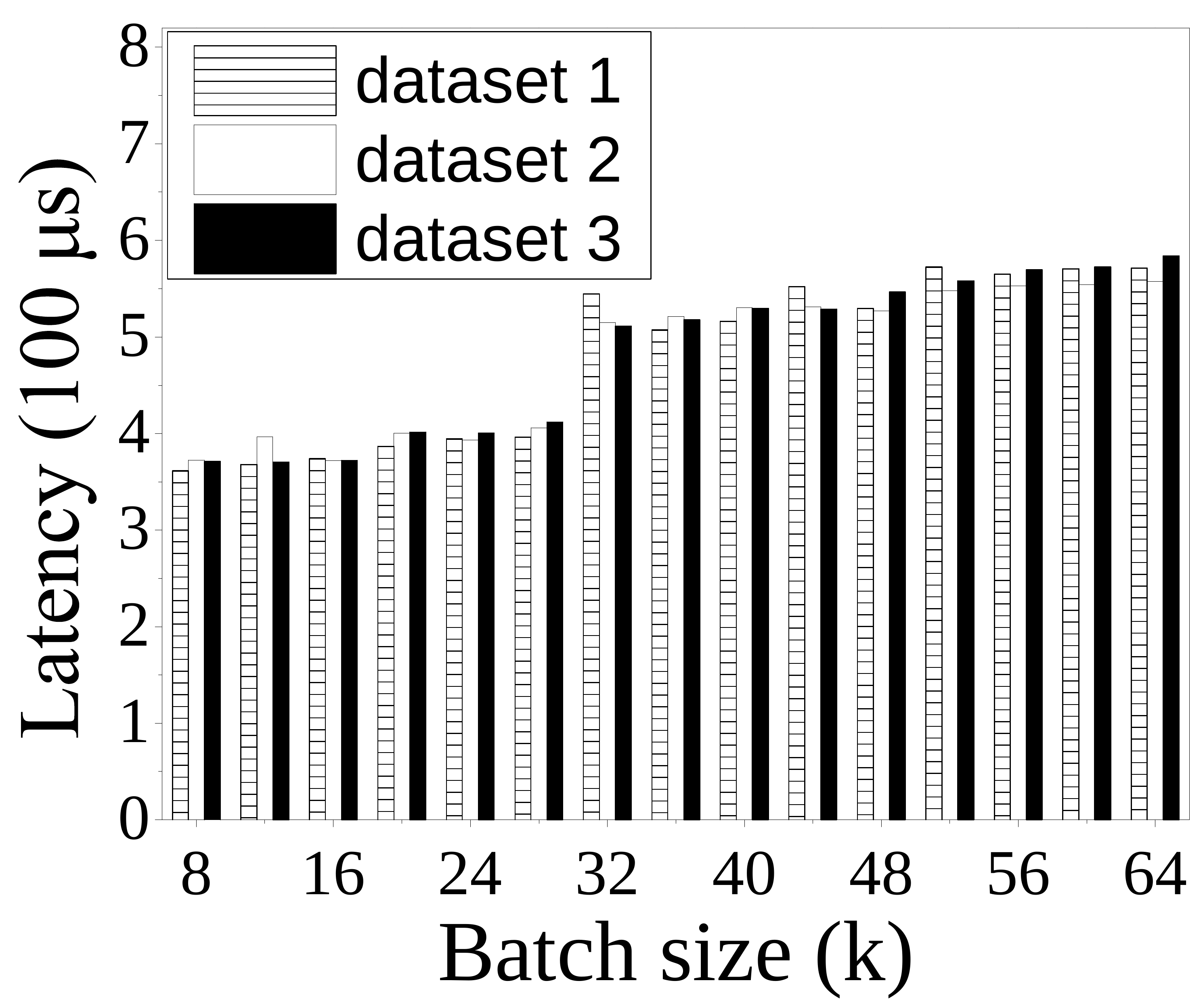}
			\caption{GPU latency vs. batch size.}
			\label{fig:gpu:exp:latency}
		\end{minipage}
\end{figure*}


\presub
\subsubsection{Skewed Workload}
For skewed workloads,
%
we performed exactly the same experiments as for the uniform workloads.
The results from these experiments are shown in Figures \ref{fig:cpu:exp:CDF-RE-ZIPF}, \ref{fig:cpu:exp:ARE-ZIPF-INS}, \ref{fig:cpu:exp:ARE-ZIPF-DEL}, and \ref{fig:cpu:exp:ARE-ZIPF-INS-DEL}.
%
%
The trends in these figures for the skewed distribution are similar to what we observed for the uniform distribution.
Therefore, for the sake of brevity, next, we concisely report the results without describing again how the experiments were conducted.

%

\noindent\textbf{Relative Error CDF:}
\emph{Our experimental results, reported in Figure \ref{fig:cpu:exp:CDF-RE-ZIPF}, show that in case of skewed workload, the percentage of items for which the relative error of our SF-sketch is less than 1\% is 74.30\%, which is 21.3, 4.6, 2.1 and 1.7 times higher than the corresponding percentages for CML, C, CM and CU sketches, respectively.
}

%



\noindent\textbf{Relative Error vs. \# of Insertions:}
\textit{Our experimental results, reported in Figure \ref{fig:cpu:exp:ARE-ZIPF-INS}, show that in case of skewed workload, the average relative error of SF-sketch is $[0.3, 2.8]$, $[3.2, 12.7]$, $[3.0, 14.8]$, and $[1.5, 2.7]$ times smaller than the average relative errors of CML, C, CM, and CU sketches, respectively.}
The converged average relative error of our SF-sketch is 2.8, 12.7, 14.8 and 2.7 times smaller than the converged average relative errors of CML, C, CM, and CU-sketch, respectively.

\noindent\textbf{Relative Error vs. \# of Deletions:}
\textit{Our experimental results, reported in Figure~\ref{fig:cpu:exp:ARE-ZIPF-DEL}, show that for skewed workload, the average relative error of SF-sketch is [2.1 to 12.7] and [1.9 to 14.8] times smaller than the average relative errors of C and CM-sketches, respectively.}
%


\noindent\textbf{Increase in Error due to Deletions:}
\textit{Our experimental results, reported in Figure~\ref{fig:cpu:exp:ARE-ZIPF-INS-DEL}, show that SF-sketch looses some accuracy due to deletions, while C and CM-sketches do not. Despite that, the average relative error of SF-sketch is still [2.1 to 12.7] times lower compared to the C-sketch and [1.9 to 14.8] times lower compared to the CM-sketch.}
%


\vspace{0.05in}
\presub
\subsubsection{Real Traffic}
\postsub
We also used real traffic to evaluate the accuracy of sketches.
As real traffic does not have deletion operations, we only show results for relative error CDF and relative error with respect to insertions.
We have 10M real traffic instances and regard the traffic with the same destination IP address to belong to the same flow.
Using this definition of a flow, there are about 230K flows in our data set, and the size distribution of flows is biased with an expected value of 8.1 and variance of 1606.
We set $d$ = 5, $w$ = 300000, $z$ = 20 in this set of experiments.

\noindent\textbf{Relative Error CDF:}
\emph{Our experimental results, reported in Figure \ref{fig:real:cdf}, show that in case of real world traffic, after a total of 10M insertions of the 230K distinct items, 99.81\% items with our SF-sketch have no error, while the percentages of items with CML, C, CM and CU sketches are 72.26\%, 79.28\%, 96.21\% and 99.06\%, respectively.
}

\noindent\textbf{Relative Error vs. \# of Insertions:}
\textit{Our experimental results, reported in Figure~\ref{fig:real:insert}, show that for our real world traffic, the average relative error of SF-sketch is [25.4 to 4246.2], [98.0 to 1341.8], [11.8 to 17.3]  and [3.5 to 4.8] times smaller than the average relative errors of CML, C, CM, and CU sketches, respectively.}

\vspace{0.1in}
\presub
\subsubsection{Sketch Parameters}
\postsub
Next, we evaluate the effect of changing the system parameters $d$ (the number of arrays) and $w$ (the number of buckets per array) on the accuracy of the sketches.
%
In each experiment to evaluate the effect of system parameters, we insert the 100K distinct items 10M times.

\textbf{Accuracy vs. \textit{w}:}
\textit{Our experimental results show that the CU-sketch requires 1.5 times more memory compared to the SF-sketch to achieve close to 1\% average relative error.}
Figure~\ref{fig:cpu:exp:opt-fixed-d} plots the average relative error by varying the number of buckets per array with $d$ fixed at $5$.
We observe from this figure that increasing the number of buckets per array reduces the average relative error for each sketch.
However, we observe that at 30K buckets per array, the average relative error of our SF-sketch reduces to a very small value of just 0.047.
On the contrary, the CU-sketch requires 50K buckets per array to achieve an average relative error of 0.049, but note that it does not support deletions.
The CML, C and CM-sketch did not achieve close to 0 average relative errors in our experiment.



\textbf{Accuracy vs. \textit{d}:}
\textit{Our experimental results show that our SF-sketch achieve an average relative error of 5.6\% using only 3 arrays whereas the CU-sketch takes 6 arrays to come close to the error of SF-sketch and achieves the average relative error of 7.1\%.}
This shows that compared to CU-sketch, at this error rate, SF-sketch takes only half as much memory.
Figure~\ref{fig:cpu:exp:opt-fixed-w} plots the average relative error by varying $d$ with $w$ fixed at $40K$.
We observe from this figure that using 6 arrays, SF-sketch achieve an average relative error of 1.9\%.
We also observe that increasing the number of arrays reduces the average relative error for all the sketches.
%
%
\balance


%


\presub
\subsection{Experiments on Update Speed and Query Speed}
\postsub
Next, we evaluate the update speed and query speed of the sketches.
For GPU platform, the query speed is the throughput, and we also need to evaluate the latency, which is measured in microseconds and quantifies the time duration between submitting a query and receiving the response.
%
%
%
We observed from our experiments that the query speeds of all the five sketches (\ie, CML, C, CM, CU, SF sketches) have little differences.
This observation was expected because the query operations of these sketches are almost the same.
For this reason, we only present the experimental results of the query speed of SF-sketch and CM-sketch.

\presub
\subsubsection{Update Speed}
\label{subsubsec:updatespeed}
\postsub
\textit{Our experimental results show that the insertion and deletion speeds of SF-sketch are $0.56$ times and $0.46$ times of those of CM-sketch, respectively.}
Our experimental results show that the CM-sketch achieves the fastest insertion and deletion speeds among CML, C, CU sketches.
Thus, we only compare our SF-sketch with CM-sketch in terms of update speed (measured in Mups, million updates per second).
Figure~\ref{fig:cpu:exp:updatespeed} plots the insertion and deletion speeds of both SF-sketch with CM-sketch.
We conducted experiments using many datasets and observed similar results.
Figure~\ref{fig:cpu:exp:updatespeed} shows plots for 10 randomly chosen datasets.
We use the same sketch parameters as used in Figure \ref{fig:cpu:exp:ARE-ZIPF-DEL}.
Our results also show that the insertion and deletion speeds of SF-sketch are nearly half of those of CM-sketch.
This is because SF-sketch needs to update two subsketches for each update.
In this experiment, we only use one CPU core to perform updates.
One can achieve much faster update speed using multiple cores or GPU.

\presub
\subsubsection{Multi-core CPU Platform}
\label{subsubsec:OnMulti-coreCPUPlatform}
%
\textit{Our experimental results show that the SF-sketch experienced a throughput gain of about 650K queries per second per thread up to 24 threads.}
Figure~\ref{fig:cpu:exp:multi-threads} plots the throughput vs. the number of threads for the SF-sketch.
We observe from this figure that SF-sketch achieves a throughput of about 1.34M queries per second with a single thread.
For this experiment, we performed 10M queries.
%
%
Using 24 threads, it achieves a throughput of about 16.3M queries per second.
We further observed that increasing the number of threads beyond 24 hardly brought about any improvement because our CPU has $6\times 2$ cores, which support $6\times2\times2=24$ threads.
This suggests that the query speed of our SF-sketch increases linearly with the number of CPU cores.
%

%



\presub
\subsubsection{GPU Platform}
%

\textit{ Our experimental results for three different data sets show that the query speed in GPU increases with the increase in the batch size.}
As shown in Figure~\ref{fig:gpu:exp:speed}, for the batch size of 20K queries, the query speed is around 50 million queries per second (Mqps).
With increase in the batch size, such as 64K queries per batch, SF-sketch reaches a query speed higher than 110 Mqps.
%


\textit{Our experimental results for three different data sets show that for SF-sketch, to reduce latency, the batch size of 28K is the most optimal for our experimental setup.}
Figure \ref{fig:gpu:exp:latency} shows that the average query latency of SF-sketch is below $410~\mu s$ for batch sizes $\leqslant 28K$.
For batch sizes $\geqslant 32k$, the latency increases to $511 \sim 584 ~\mu s$.
%
%

	\presec
\section{Conclusion} \postsec
\label{sec:conclusion}
%
In this paper, we proposed a new sketch for data streams, namely the SF-sketch, which achieves up to 33.1 times higher accuracy compared to the CM-sketch while keeping the query speed as fast as the CM-sketch.
The key idea behind our proposed SF-sketch is to maintain two separate sketches, one is called Slim-subsketch and the other is called Fat-subsketch.
We send only the Slim-subsketch to the remote collector to achieve accurate and fast query, while the Fat-subsketch enables the Slim-subsketch to achieve high query accuracy.
To evaluate and compare the performance of our proposed SF-sketch, we conducted extensive experiments on multi-core CPU and GPU platforms.
Our experimental results show that our SF-sketch significantly outperforms the-state-of-the-art in terms of accuracy.

	\balance
	
	\vfill
	\bibliographystyle{abbrv}
	\bibliography{reference}

\begin{thebibliography}{10}

\bibitem{memcached}
{M}emcached - a distributed memory object caching system.
\newblock \url{http://memcached.org}.

\bibitem{opensource}
Source code of {SF} sketches.
\newblock \url{https://github.com/paper2017/SF-sketch}.

\bibitem{sdm10sketch}
C.~C. Aggarwal and S.~Y. Philip.
\newblock On classification of high-cardinality data streams.
\newblock In {\em SDM}, volume~10, pages 802--813. SIAM, 2010.

\bibitem{dynamicCF}
J.~Aguilar-Saborit, P.~Trancoso, V.~Muntes-Mulero, and J.-L. Larriba-Pey.
\newblock Dynamic count filters.
\newblock {\em ACM SIGMOD Record}, pages 26--32, 2006.

\bibitem{countsketch}
M.~Charikar, K.~Chen, and M.~Farach-Colton.
\newblock Finding frequent items in data streams.
\newblock In {\em Automata, Languages and Programming}. Springer, 2002.

\bibitem{chen2010tracking}
A.~Chen, Y.~Jin, J.~Cao, and L.~E. Li.
\newblock Tracking long duration flows in network traffic.
\newblock In {\em \Proc IEEE INFOCOM}, 2010.

\bibitem{spectralBF}
S.~Cohen and Y.~Matias.
\newblock {S}pectral bloom filters.
\newblock In {\em \Proc ACM SIGMOD}, pages 241--252, 2003.

\bibitem{YCSB}
B.~F. Cooper, A.~Silberstein, and et~al.
\newblock Benchmarking cloud serving systems with {YCSB}.
\newblock In {\em \Proc ACM SOCC}, 2010.

\bibitem{cormode2005sketching}
G.~Cormode and M.~Garofalakis.
\newblock Sketching streams through the net: Distributed approximate query
  tracking.
\newblock In {\em Proceedings of the 31st international conference on Very
  large data bases}, pages 13--24. VLDB Endowment, 2005.

\bibitem{cormode2008finding}
G.~Cormode and M.~Hadjieleftheriou.
\newblock Finding frequent items in data streams.
\newblock {\em Proceedings of the VLDB Endowment}, 1(2):1530--1541, 2008.

\bibitem{cmsketch}
G.~Cormode and S.~Muthukrishnan.
\newblock An improved data stream summary: the count-min sketch and its
  applications.
\newblock {\em Journal of Algorithms}, 55(1):58--75, 2005.

\bibitem{durme2009streaming}
B.~V. Durme and A.~Lall.
\newblock Streaming pointwise mutual information.
\newblock In {\em Advances in Neural Information Processing Systems}, pages
  1892--1900, 2009.

\bibitem{cusketch}
C.~Estan and G.~Varghese.
\newblock New directions in traffic measurement and accounting.
\newblock {\em ACM SIGMCOMM CCR}, 32(4), 2002.

\bibitem{gilbert2007one}
A.~C. Gilbert, M.~J. Strauss, J.~A. Tropp, and R.~Vershynin.
\newblock One sketch for all: fast algorithms for compressed sensing.
\newblock In {\em Proceedings of ACM symposium on Theory of computing}, 2007.

\bibitem{NLPsketch11}
A.~Goyal and H.~Daum{\'e}~III.
\newblock Approximate scalable bounded space sketch for large data nlp.
\newblock In {\em \Proc EMNLP}, 2011.

\bibitem{goyalIII}
A.~Goyal and H.~Daum{\'e}~III.
\newblock Approximate scalable bounded space sketch for large data nlp.
\newblock In {\em \Proc EMNLP}, 2011.

\bibitem{cubest}
A.~Goyal, H.~Daum{\'e}~III, and G.~Cormode.
\newblock Sketch algorithms for estimating point queries in nlp.
\newblock In {\em \Proc EMNLP-CoNLL}, pages 1093--1103. Association for
  Computational Linguistics, 2012.

\bibitem{NLPstreaming09}
A.~Goyal, H.~Daum{\'e}~III, and S.~Venkatasubramanian.
\newblock Streaming for large scale nlp: Language modeling.
\newblock In {\em \Proc NAACL}. Association for Computational Linguistics,
  2009.

\bibitem{remotecollectorNSDI}
N.~Handigol, B.~Heller, and et~al.
\newblock I know what your packet did last hop: Using packet histories to
  troubleshoot networks.
\newblock In {\em NSDI}, volume~14, pages 71--85, 2014.

\bibitem{li1sketch}
P.~Li, K.~W. Church, and T.~J. Hastie.
\newblock One sketch for all: Theory and application of conditional random
  sampling.
\newblock In {\em Advances in Neural Information Processing Systems}, pages
  953--960, 2009.

\bibitem{flowradar}
Y.~Li, R.~Miao, C.~Kim, and M.~Yu.
\newblock Flowradar: a better netflow for data centers.
\newblock In {\em \Proc USENIX NSDI}, pages 311--324, 2016.

\bibitem{GPU-GAMT}
Y.~Li, D.~Zhang, A.~X. Liu, and J.~Zheng.
\newblock {GAMT}: a fast and scalable ip lookup engine for gpu-based software
  routers.
\newblock In {\em \Proc IEEE/ACM ANCS}, 2013.

\bibitem{sigsketch}
Z.~Liu, A.~Manousis, and et~al.
\newblock One sketch to rule them all: Rethinking network flow monitoring with
  univmon.
\newblock In {\em \Proc ACM SIGCOMM}, 2016.

\bibitem{CUDA-Best}
NVIDIA Corporation.
\newblock {\em NVIDIA CUDA C Best Practices Guide, Version 5.0}, Oct. 2012.

\bibitem{cmlog}
G.~Pitel and G.~Fouquier.
\newblock Count-min-log sketch: Approximately counting with approximate
  counters.
\newblock {\em arXiv preprint arXiv:1502.04885}, 2015.

\bibitem{asketch}
P.~Roy, A.~Khan, and G.~Alonso.
\newblock Augmented sketch: Faster and more accurate stream processing.
\newblock In {\em \Proc ACM SIGMOD}, 2016.

\bibitem{sketchtheorysigmod}
F.~Rusu and A.~Dobra.
\newblock Statistical analysis of sketch estimators.
\newblock In {\em Proceedings of the 2007 ACM SIGMOD international conference
  on Management of data}, pages 187--198. ACM, 2007.

\bibitem{sketchtheorytods}
F.~Rusu and A.~Dobra.
\newblock Sketches for size of join estimation.
\newblock {\em ACM Transactions on Database Systems (TODS)}, 33(3):15, 2008.

\bibitem{icde09fcm}
D.~Thomas, R.~Bordawekar, and et~al.
\newblock On efficient query processing of stream counts on the cell processor.
\newblock In {\em \Proc IEEE ICDE}, 2009.

\bibitem{GPU-NDN}
Y.~Wang, Y.~Zu, and et~al.
\newblock Wire speed name lookup: A gpu-based approach.
\newblock In {\em \Proc USENIX NSDI}, pages 199--212, 2013.

\bibitem{shbf}
T.~Yang, A.~X. Liu, M.~Shahzad, Y.~Zhong, Q.~Fu, Z.~Li, G.~Xie, and X.~Li.
\newblock A shifting bloom filter framework for set queries.
\newblock In {\em \Proc VLDB}, 2016.

\bibitem{short}
T.~Yang, L.~Liu, Y.~Yan, M.~Shahzad, Y.~Shen, X.~Li, B.~Cui, and G.~Xie.
\newblock Sf-sketch: A fast, accurate, and memory efficient data structure to
  store frequencies of data items.
\newblock In {\em \em Proceedings of the ICDE}, 2017.

\bibitem{sail}
T.~Yang, G.~Xie, Y.~Li, Q.~Fu, A.~X. Liu, Q.~Li, and L.~Mathy.
\newblock Guarantee ip lookup performance with fib explosion.
\newblock In {\em ACM Conference on SIGCOMM}, pages 39--50, 2015.

\bibitem{remotecollectorsigcomm}
Y.~Zhu, N.~Kang, and et~al.
\newblock Packet-level telemetry in large datacenter networks.
\newblock In {\em \Proc ACM SIGCOMM}, 2015.

\end{thebibliography}
    \appendix

\section*{Appendix I: Proof of A$_\textbf{i}$[j] $\leqslant$ max$_{\textbf{k=1}}^\textbf{z}$B$_\textbf{i}$[j][k]} \label{appdx:ABCounterValuesRelationship}
Here, we use induction to prove that for SF$_\text{F}$-sketch, the value of each counter in the Slim-subsketch is always less than or equal to the value of the largest counter in the corresponding bucket of the Fat-subsketch.
Before delving into the proof, we first define some notations.
Let $A^n_i[j]$ represent the value of the counter in the $j^{\text{th}}$ bucket of the $i^{\text{th}}$ array in the Slim-subsketch after $n$ updates, where $1\leqslant i\leqslant d$, $1\leqslant j\leqslant w$, and $n \in \mathbb{N}^{0}$.
Similarly, let $B^n_i[j][k]$ represent the value of the $k^{\text{th}}$ counter in the $j^{\text{th}}$ bucket of the $i^{\text{th}}$ array in the fat-subsketch after $n$ updates, where $1\leqslant i\leqslant d$, $1\leqslant j\leqslant w$, $1\leqslant k\leqslant z$, and $n \in \mathbb{N}^{0}$.
As querying does not change any counter in either Fat-subsketch or Slim-subsketch, we only consider insertion and deletion operations for this proof.

\begin{thm}
\label{thm:ABCounterValuesRelationship}
Each counter in the Slim-subsketch of the SF$_\text{F}$-sketch is always less than or equal to the largest counter in the corresponding bucket of the Fat-subsketch regardless of the number of insertion and deletion operations performed.
Formally, $\forall i\in[1,d], j\in[1,w]$, and $n \in \mathbb{N}^{0}$,
\begin{equation}
\label{equ_a:not_bigger}
A^{n}_i[j] \leqslant \max_{k=1}^{z}B^{n}_i[j][k]
\end{equation}
\end{thm}

\begin{proof}
We prove this theorem by induction on $n$, \ie, the number of updates (insertions and deletions).

\noindent\textbf{Base case:}
The base case occurs when $n=0$, \ie, no update has yet been done and all counters in both subsketches are currently 0.
Hence, $\forall i\in[1,d]$ and $j\in[1,w]$, $A^{0}_i[j] = \max_{k=1}^{z}B^{0}_i[j][k]=0$.
Thus, the base case of $n=0$ satisfies Equation \eqref{equ_a:not_bigger}.

\noindent\textbf{Induction hypothesis:}
Let for $n = x$, Equation \eqref{equ_a:not_bigger} holds true, \ie,
$\forall i\in[1,d], j\in[1,w]$, and $n \leqslant x$, $A^{x}_i[j] \leqslant \max_{k=1}^{z}B^{x}_i[j][k]$.
Our inductive hypothesis is that if Equation \eqref{equ_a:not_bigger} holds true for $\forall i\in[1,d]$ and $j\in[1,w]$ when $n\leqslant x$, then it also holds true when $n=x+1$, \ie, $A^{x+1}_i[j] \leqslant \max_{k=1}^{z}B^{x+1}_i[j][k]$.

\noindent\textbf{Induction step:}
To prove that the induction hypothesis is correct, we pick an arbitrary counter $A_i[j]$ in the Slim-subsketch, whose current value after $x$ updates is $A^x_i[j]$.
We show that the induction hypothesis is true for this arbitrary counter.
As our choice of the counter is arbitrary, this in turn implies that all corresponding counters of the Slim- and Fat-subsketches satisfy the inductive hypothesis.

Let the item that is inserted/deleted in the $x+1^{\text{st}}$ update is represented by $e$.
There are following two cases to consider when performing the $x+1^{\text{st}}$ update: $j\neq h_i(e)$ and $j=h_i(e)$.
We consider them one by one.

\noindent\underline{\textbf{\textit{Case 1:}} $j\neq h_i(e)$}
This case implies that in inserting or deleting $e$, neither the value of the counter $A_i[j]$ is modified nor the value of any counter in the bucket $B_i[j]$ is modified.
Consequently, $A^{x+1}_i[j]=A^{x}_i[j]$ and $\forall k\in[1,z], B^{x+1}_i[j][k]=B^{x}_i[j][k]$.
Thus, $A^{x+1}_i[j]\leqslant \max_{k=1}^z B^{x+1}_i[j][k]$.
The inductive hypothesis holds true for Case 1.

\noindent\underline{\textbf{\textit{Case 2:}} $j= h_i(e)$}
This case implies that in inserting or deleting $e$, the value of the counter $A_i[j]$ and the value of $\max_{k=1}^z B_i[j]$ may be modified.
We consider the insertion and deletion operations separately.

\noindent\textbf{\textit{Insertion:}}
For insertion, there are two cases that we need to consider: when $A^{x}_i[j] < \max_{k=1}^{z}B^{x}_i[j][k]$ and when $A^{x}_i[j] = \max_{k=1}^{z}B^{x}_i[j][k]$.

For the first case of insertion, when $A^{x}_i[j] < \max_{k=1}^{z}B^{x}_i[j][k]$, after incrementing the value of a counter in the bucket $B_i[j]$ for the $x+1^\text{st}$ update, clearly the following inequality holds: $\max_{k=1}^{z}B^{x}_i[j][k] \leqslant \max_{k=1}^{z}B^{x+1}_i[j][k] $.
As the value of $A_i[j]$ can at most increment by 1 during an insertion, the following inequality also holds $A^{x+1}_i[j] \leqslant \max_{k=1}^{z}B^{x}_i[j][k]$.
Combining these two inequalities, we get $A^{x+1}_i[j] \leqslant \max_{k=1}^{z}B^{x}_i[j][k] \leqslant \max_{k=1}^{z}B^{x+1}_i[j][k]$.
Removing the middle term, we get $A^{x+1}_i[j] \leqslant \max_{k=1}^{z}B^{x+1}_i[j][k]$.
Thus, the inductive hypothesis holds true for this first case of insertion.

For the second case of insertion, when $A^{x}_i[j] = \max_{k=1}^{z}B^{x}_i[j][k]$, we use proof by contradiction to prove that the induction hypothesis is true.
Assume that the induction hypothesis is not true for this second case of insertion, which means that the inequality $A^{x+1}_i[j] > \max_{k=1}^{z}B^{x+1}_i[j][k]$ is true.
For this inequality to be true, according to our insertion algorithm of SF$_\text{F}$-sketch, following two conditions must be met at the same time: $\max_{k=1}^{z}B^{x+1}_i[j][k] = \max_{k=1}^{z}B^{x}_i[j][k]$ and $A^{x+1}_i[j] = A^{x}_i[j]+1$.
Suppose the first condition is satisfied, \ie, $\max_{k=1}^{z}B^{x+1}_i[j][k] = \max_{k=1}^{z}B^{x}_i[j][k]$, which implies that $B^{x}_i[j][f_i(e)] < \max_{k=1}^{z}B^{x}_i[j][k]$.
Furthermore, $B^{x}_i[j][f_i(e)] = B^{x+1}_i[j][f_i(e)] - 1$.
We get

\begin{IEEEeqnarray}{CrCl}
&B^{x+1}_i[j][f_i(e)] -1 &< &\max_{k=1}^{z}B^{x}_i[j][k]\nonumber\\
&\Rightarrow B^{x+1}_i[j][f_i(e)] &\leqslant &\max_{k=1}^{z}B^{x}_i[j][k]
\label{equ:implicationCond1Ins}
\end{IEEEeqnarray}

Suppose the second condition is also satisfied, \ie, $A^{x+1}_i[j] = A^{x}_i[j]+1$, which implies that the counter value $A^x_i[j]$ must be minimum among all counters in the Slim-subsketch to which the item $e$ maps, \ie, $A_i^{x}[j] = \min_{i'=1}^{d} A_{i'}^{x}[h_{i'}(e)]$, and this counter value $A^x_i[j]$ must also be smaller than the smallest value among all counters in the Fat-subsketch to which the item $e$ maps, \ie, $A_i^{x}[j] < \min_{i'=1}^{d}B^{x+1}_{i'}[h_{i'}(e)][f_{i'}(e)]$.
As $\min_{i'=1}^{d}B^{x+1}_{i'}[h_{i'}(e)][f_{i'}(e)]\leqslant B^{x+1}_{i}[j][f_{i}(e)]$, we get $A_i^{x}[j]<B^{x+1}_{i}[j][f_{i}(e)]$.
Applying the result from Equation \eqref{equ:implicationCond1Ins}, we get $A_i^{x}[j]<\max_{k=1}^{z}B^{x}_i[j][k]$.
This contradicts with the condition of this second case of insertion, which states that $A^{x}_i[j] = \max_{k=1}^{z}B^{x}_i[j][k]$.
Thus, the assumption that the induction hypothesis is not true for this second case of insertion,\ie, $A^{x+1}_i[j] > \max_{k=1}^{z}B^{x+1}_i[j][k]$ is not correct.
This implies that, $A^{x+1}_i[j] \leqslant \max_{k=1}^{z}B^{x+1}_i[j][k]$, and the induction hypothesis holds true for this second case of insertion.

\noindent\textbf{\textit{Deletion:}}
For deletion, there are again two cases that we need to consider just like for insertion: when $A^{x}_i[j] < \max_{k=1}^{z}B^{x}_i[j][k]$ and when $A^{x}_i[j] = \max_{k=1}^{z}B^{x}_i[j][k]$.

For the first case of deletion, when $A^{x}_i[j] < \max_{k=1}^{z}B^{x}_i[j][k]$, as the value of $\max_{k=1}^{z}B_i[j][k]$ can at most decrement by 1 during a deletion, the following inequality holds: $A^{x}_i[j]\leqslant \max_{k=1}^{z}B^{x+1}_i[j][k]$.
Similarly, as the value of $A_i[j]$ can at most decrement by 1 during a deletion, the following inequality also holds $A^{x+1}_i[j] \leqslant A^{x}_i[j]$.
Combining these two inequalities, we get $A^{x+1}_i[j] \leqslant A^{x}_i[j] \leqslant \max_{k=1}^{z}B^{x+1}_i[j][k]$.
Removing the middle term, we get $A^{x+1}_i[j] \leqslant \max_{k=1}^{z}B^{x+1}_i[j][k]$.
Thus, the inductive hypothesis holds true for this first case of deletion

For the second case of deletion, when $A^{x}_i[j] = \max_{k=1}^{z}B^{x}_i[j][k]$, we use proof by contradiction to prove that the induction hypothesis is true.
Assume that the induction hypothesis is not true for this second case of insertion, which means that the inequality $A^{x+1}_i[j] > \max_{k=1}^{z}B^{x+1}_i[j][k]$ is true.
For this inequality to be true, according to our deletion algorithm of SF$_\text{F}$-sketch, following two conditions must be met at the same time: $\max_{k=1}^{z}B^{x+1}_i[j][k] = \max_{k=1}^{z}B^{x}_i[j][k] - 1$ and $A^{x+1}_i[j] = A^{x}_i[j]$.
Suppose the first condition is satisfied, \ie, $\max_{k=1}^{z}B^{x+1}_i[j][k] = \max_{k=1}^{z}B^{x}_i[j][k] - 1$, which implies the following.

\begin{IEEEeqnarray}{CrCl}
&\max_{k=1}^{z}B^{x+1}_i[j][k] &<& \max_{k=1}^{z}B^{x}_i[j][k]
\label{equ:implicationCond1Del}
\end{IEEEeqnarray}

Suppose the second condition is also satisfied, \ie, $A^{x+1}_i[j] = A^{x}_i[j]$, which implies that $A_i^{x}[j] \leqslant \max_{k=1}^{z}B^{x+1}_i[j][k]$.
Substituting the value of $\max_{k=1}^{z}B^{x+1}_i[j][k]$ from Equation \eqref{equ:implicationCond1Del}, we get $A_i^{x}[j] < \max_{k=1}^{z}B^{x}_i[j][k]$.
This contradicts with the condition of this second case of insertion, which states that $A^{x}_i[j] = \max_{k=1}^{z}B^{x}_i[j][k]$.
Thus, the assumption that the induction hypothesis is not true for this second case of deletion,\ie, $A^{x+1}_i[j] > \max_{k=1}^{z}B^{x+1}_i[j][k]$ is not correct.
This implies that, $A^{x+1}_i[j] \leqslant \max_{k=1}^{z}B^{x+1}_i[j][k]$, and the induction hypothesis holds true for this second case of deletion.
\end{proof}

\section*{Appendix II: Proof of No Under-estimation Error}\label{appdx:no_under}
Here, we use induction to prove that SF$_\text{F}$-sketch does not incur any under-estimation error.
Before delving into the proof, we first define some notations.
Similar to the notations defined in the Appendix \ref{appdx:ABCounterValuesRelationship}, let $A^n_i[j]$ represent the value of the counter in the $j^{\text{th}}$ bucket of the $i^{\text{th}}$ array in the Slim-subsketch after $n$ updates, where $1\leqslant i\leqslant d$, $1\leqslant j\leqslant w$, and $n \in \mathbb{N}^{0}$.
Similarly, let $B^n_i[j][k]$ represent the value of the $k^{\text{th}}$ counter in the $j^{\text{th}}$ bucket of the $i^{\text{th}}$ array in the fat-subsketch after $n$ updates, where $1\leqslant i\leqslant d$, $1\leqslant j\leqslant w$, $1\leqslant k\leqslant z$, and $n \in \mathbb{N}^{0}$.
Furthermore, let $SQ^{n}(e)$ represent estimate of the frequency of an arbitrary item $e$ calculated from the Slim-subsketch after $n$ updates.
Similarly, let $FQ^{n}(e)$ represent estimate of the frequency of an arbitrary item $e$ calculated from the Fat-subsketch after $n$ updates.
As Fat-subsketch of the SF$_\text{F}$-sketch behaves exactly like a CM-sketch for any given item $e$, we directly use the conclusion from \cite{cmsketch} that CM-sketch does not suffer from underestimation error, \ie, for any arbitrary item $e$ and $\forall n \in \mathbb{N}^{0}, FQ^{n}(e) \geqslant f^{n}_{(e)}$, where $f^{n}_{(e)}$ represents the actual frequency of item $e$ after $n$ updates.
As querying does not change any counter in either Fat-subsketch or Slim-subsketch, we only consider insertion and deletion operations (collectively called updates) in this proof.

\begin{thm}
\label{thm_b:no_under}
For updates consisting of insertions and deletions, the SF$_\text{F}$-sketch does not incur underestimation error.
Formally, for any arbitrary item $e$ and $\forall n \in \mathbb{N}^{0}$

\begin{equation}
\label{equ_b:no_under}
SQ^{n}(e) \geqslant f^{n}_{(e)}
\end{equation}
\end{thm}

\begin{proof}
We prove this theorem by induction on $n$, \ie, the number of updates (insertions and deletions).

\noindent\textbf{Base case:}
The base case occurs when $n=0$, \ie, no update has yet been done and all counters in both sub-sketches are currently 0.
Thus, for any item, the frequency returned by the Slim-subsketch at this point will be 0, \ie for any arbitrary item $e$, $SQ^0(e)=0$.
As no items have yet been inserted, therefore, $f^0_{(e)}=0$.
Therefore, $SQ^0(e)= f^0_{(e)}$.
Thus, the base case of $k=0$ satisfies Equation \eqref{equ_b:no_under}.

\noindent\textbf{Induction hypothesis:}
Let for $n = x$, Equation \eqref{equ_b:no_under} holds true, \ie, for any arbitrary item $e$, $SQ^{x}(e) \geqslant f^{x}_{(e)}$.
Our inductive hypothesis is that if Equation \eqref{equ_b:no_under} holds true for any arbitrary item $e$ when $n\leqslant x$, then it also holds true for that item when $n=x+1$, \ie, for item $e$, $SQ^{x+1}(e) \geqslant f^{x+1}_{(e)}$.

\noindent\textbf{Induction step:}
We use $n=x+1$ in the inductive step.
The value of $SQ^{x}(e)$ is given as below.
\begin{equation}
\begin{aligned}
SQ^{x}(e) = A^{x}_{i'}[h_{i'}(e)]
\end{aligned}
\label{equ_b:thm_1}
\end{equation}
where $i'\in [1,d]$ such that $A^{x}_{i'}[h_{i'}(e)] = \min_{i=1}^{d}A^{x}_i[h_i(e)]$.
Suppose the item that we are going to insert/delete in the $x+1^{\text{st}}$ update is $e'$, where $e'$ can also be any arbitrary item.
We choose to use $e'$ in the induction step instead of $e$ to make the inductive step generic.
Note that the inductive step holds when $e'=e$ as well as when $e'\neq e$.
Next, we apply the inductive step separately to the two cases: 1) item $e'$ is inserted; 2) item $e'$ is deleted.

\noindent\textit{\textbf{Insertion:}}
When inserting the item $e'$, there are two cases.
The first case occurs when $h_{i'}(e)=h_{i'}(e')$ regardless of whether $e=e'$ or $e\neq e'$.
The second case occurs when $h_{i'}(e)\neq h_{i'}(e')$.
Note that this second case occurs only when $e\neq e'$.

For the first case, where $h_{i'}(e)=h_{i'}(e')$, the counter $A_{i'}[h_{i'}(e)]$ may be incremented if $A^{x}_{i'}[h_{i'}(e)] < \min_{i=1}^{d}B^{x+1}_i[h_i(e')][f_i(e')]$.
In this case,
    \begin{equation*}
    \begin{aligned}
        A^{x+1}_{i'}[h_{i'}(e)]&=A^{x}_{i'}[h_{i'}(e)]+1\\
        &= SQ^x(e)+1\\
        &\geqslant f^x_{(e)}+1
    \end{aligned}
    \end{equation*}
If $e\neq e'$, then  $f^x_{(e)}+1>f^{x+1}_{(e)}$, otherwise $f^x_{(e)}+1=f^{x+1}_{(e)}$.
Combining these two, we get $f^x_{(e)}+1\geqslant f^{x+1}_{(e)}$.
Substituting this inequality in the equation above, we get
    \begin{equation*}
    \begin{aligned}
        A^{x+1}_{i'}[h_{i'}(e)]\geqslant f^x_{(e)}+1\geqslant f^{x+1}_{(e)}
    \end{aligned}
    \end{equation*}
As $A^{x+1}_{i'}[h_{i'}(e)]=SQ^{x+1}(e)$, we get
    \begin{equation*}
    \begin{aligned}
        SQ^{x+1}(e)\geqslant f^{x+1}_{(e)}
    \end{aligned}
    \end{equation*}
Thus, the inductive hypothesis holds true for the first case.

For the second case, where $h_{i'}(e)\neq h_{i'}(e')$, the counter $A_{i'}[h_{i'}(e)]$ stays unchanged, \ie, $A^{x+1}_{i'}[h_{i'}(e)] = A^{x}_{i'}[h_{i'}(e)]$, which means that $SQ^{x+1}(e) = SQ^x(e)$.
This further implies that $SQ^{x+1}(e) \geqslant f^{x}_{(e)}$.
As for this case, $e\neq e'$, $f^x_{(e)}=f^{x+1}_{(e)}$, we get $SQ^{x+1}(e) \geqslant f^{x+1}_{(e)}$.
Thus, the inductive hypothesis holds true for the second case as well.

\noindent\textit{\textbf{Deletion:}}
According to the deletion operation of the SF$_\text{F}$-sketch, a deletion after $x$ updates may cause a counter $A_i[j]$ in Slim-subsketch to either stay unchanged or be set to $\max_{k=1}^{z}B^{x+1}_{i}[h_i(e)][k]$.
For each $i\in[1,d]$, if $h_{i}(e)=h_{i}(e')$ and $A^{x}_{i}[h_{i}(e)] > \max_{k=1}^{z}B^{x+1}_i[h_i(e)][k]$, we have
    \begin{equation}
    \begin{aligned}
        A^{x+1}_{i}[h_{i}(e)]&=\max_{k=1}^{z}B^{x+1}_i[h_i(e)][k]\\
        &\geqslant B^{x+1}_i[h_i(e)][f_i(e)]\\
        &\geqslant \min_{i'=1}^{d}B^{x+1}_{i'}[h_{i'}(e)][f_{i'}(e)]\\
        &= FQ^{x+1}(e)\\
        & \geqslant f^{x+1}_{(e)}
    \end{aligned}
    \label{equ_b:thm_del_h}
    \end{equation}
    Otherwise, we have:
    \begin{equation}
    \begin{aligned}
        A^{x+1}_{i}[h_{i}(e)] &= A^{x}_{i}[h_{i}(e)] \geqslant \min_{i'=1}^{d}A^{x}_{i'}[h_{i'}(e)]\\
        & = SQ^x(e) \geqslant f^x_{(e)} \geqslant f^{x+1}_{(e)}
    \end{aligned}
    \label{equ_b:thm_del_t}
    \end{equation}
Thus, each counter to which $e$ maps stays larger than the actual frequency of $e$ on deletion after $x$ updates, which in turn implies that $SQ^{x+1}(e)\geqslant f^{x+1}_{(e)}$ and the inductive hypothesis holds true.
\end{proof}


\end{document}